\documentclass[11pt]{article}
\usepackage{amssymb,amsfonts,amsmath,amsthm}%
\usepackage{graphicx,color}%
\usepackage{hyperref}%
\definecolor{gray}{rgb}{0.1,0.1,.1}
\hypersetup{
    colorlinks=true,       
    linkcolor=gray,          
    citecolor=gray      
}
\theoremstyle{plain}%
\newtheorem{theorem}{Theorem}[]%
\newtheorem{corollary}[theorem]{Corollary}%
\newtheorem{lemma}[theorem]{Lemma}%
\newtheorem{proposition}[theorem]{Proposition}%
\newtheorem{assumption}[theorem]{Assumption}%
\newtheorem{remark}[theorem]{Remark}%
\newcommand{\figdraft}{false}%
\newcommand{\figfile}[1]{#1}%
\definecolor{colorGreen}{rgb}{0.,0.8,0}
\definecolor{colorRed}{rgb}{0.8,0.,0}
\definecolor{colorBlue}{rgb}{0.,0.,0.8}
\newenvironment{mhchange}{%
}{
}%
\DeclareMathOperator{\xspan}{\mathrm{span}}
\DeclareMathOperator{\supp}{\mathrm{supp}}

\newcommand{\iu}{\mathtt{i}}
\newcommand{\mhexp}[1]{{{\mathtt{e}}^{#1}}}

\newcommand{\sgn}{\mathrm{sgn}}

\newcommand{\fspace}[1]{{\mathsf{#1}}}
\newcommand{\fspaceL}{\fspace{L}}

\newcommand{\fspaceC}{\fspace{C}}
\newcommand{\fspaceW}{\fspace{W}}

\newcommand{\Rset}{{\mathbb{R}}}

\newcommand{\oointerval}[2]{(#1,\,#2)}%
\newcommand{\ccinterval}[2]{[#1,\,#2]}%

\newcommand{\DO}[1]{{O\at{#1}}}

\newcommand{\nDO}[1]{{O\nat{#1}}}

\newlength{\mhpicDwidth}
\newlength{\mhpicDvsep}
\newlength{\mhpicDhsep}

\newlength{\mhpicPwidth}
\newlength{\mhpicPvsep}
\newlength{\mhpicPhsep}

\newlength{\mhpicWhsep}

\newcommand{\pair}[2]{{\left({#1},\,{#2}\right)}}

\newcommand{\npair}[2]{{({#1},\,{#2})}}

\newcommand{\at}[1]{{\left({#1}\right)}}
\newcommand{\nat}[1]{(#1)}
\newcommand{\bat}[1]{{\big(#1\big)}}
\newcommand{\Bat}[1]{{\Big(#1\Big)}}

\newcommand{\bigpar}{\par\quad\newline\noindent}

\newcommand{\norm}[1]{\|{#1}\|}
\newcommand{\abs}[1]{\left|{#1}\right|}
\newcommand{\babs}[1]{\big|{#1}\big|}
\newcommand{\nabs}[1]{|{#1}|}

\newcommand{\dint}[1]{\,\mathrm{d}#1}

\newcommand{\Om}{{\Omega}}
\newcommand{\al}{{\alpha}}
\newcommand{\be}{{\beta}}
\newcommand{\ga}{{\gamma}}

\newcommand{\la}{{\lambda}}

\newcommand{\calA}{\mathcal{A}}

\newcommand{\calF}{\mathcal{F}}
\newcommand{\calG}{\mathcal{G}}

\newcommand{\calL}{\mathcal{L}}
\newcommand{\calM}{\mathcal{M}}

\newcommand{\calP}{\mathcal{P}}

\newcommand{\calS}{\mathcal{S}}
\newcommand{\calT}{\mathcal{T}}

\def\DPsi{{C_\Psi}}

\usepackage[%
a4paper,%
nohead,%
twoside,
ignorefoot,
scale=0.85,
bindingoffset=1.25cm
]{geometry}%
\usepackage{float}%
\begin{document}%
%
%
\title{Subsonic phase transition waves \\
in bistable lattice models with small spinodal region}%
\date{\today}%
\author{ %
Michael Herrmann\footnote{
{\tt{michael.herrmann@math.uni-sb.de}}, Universit\"at des Saarlandes
}
,\;\;\;%
Karsten Matthies\footnote{
{\tt{k.matthies@maths.bath.ac.uk}}, University of Bath
}
,\;\;\;%
Hartmut Schwetlick\footnote{
{\tt{h.schwetlick@maths.bath.ac.uk}}, University of Bath
}
,\;\;\;%
Johannes Zimmer\footnote{
{\tt{zimmer@maths.bath.ac.uk}}, University of Bath
}
}
\maketitle
%
%
%
\begin{abstract}%
Phase transitions waves in atomic chains with double-well potential play a fundamental role in materials science,
but very little is known about their mathematical properties. In particular, the only available results about waves
with large amplitudes concern chains with piecewise-quadratic pair potential. In this paper we consider
perturbations of a bi-quadratic potential and prove that the corresponding three-parameter family of waves
persists as long as the perturbation is small and localised with respect to the strain variable.  As a standard 
Lyapunov-Schmidt reduction cannot be used due to the presence of an essential spectrum, 
we characterise the perturbation of the wave as a fixed point of a nonlinear and nonlocal operator and show that this
operator is contractive in a small ball in a suitable function space. Moreover, we derive a uniqueness result for
phase transition waves with certain properties and discuss the kinetic relation.
\end{abstract}%
%
%
\quad\newline\noindent%
\begin{minipage}[t]{0.15\textwidth}%
Keywords: %
\end{minipage}%
\begin{minipage}[t]{0.8\textwidth}%
\emph{phase transitions in lattices}, \emph{kinetic relations},\\
\emph{heteroclinic travelling waves in Fermi-Pasta-Ulam chains}%
\end{minipage}%
\medskip
\newline\noindent
\begin{minipage}[t]{0.15\textwidth}%
MSC (2010): %
\end{minipage}%
37K60, 
47H10,  
74J30,  
82C26   
\begin{minipage}[t]{0.8\textwidth}%
\end{minipage}%

%
%
%
%
%
\section{Introduction}
%

Many standard models in one-dimensional discrete elasticity describe the motion in atomic chains with
nearest neighbour interactions. The corresponding equation of motion reads
\begin{align}
\label{eq:FPUintro}
\ddot u_j (t) = \Phi'(u_{j+1}(t) -u_{j}(t) )-\Phi'(u_{j}(t)-u_{j-1}(t))\,,
\end{align}
where $\Phi$ is the interaction potential and $u_j$ denotes the displacement of particle $j$ at time $t$.
\par
Of particular importance is the case of non-convex $\Phi$, because then \eqref{eq:FPUintro}
provides a simple dynamical model for martensitic phase transitions.  In this context,
a propagating interface can be described by a \emph{phase transition wave}, which is a travelling wave that moves with subsonic speed and is heteroclinic as it connects periodic oscillations in different wells of $\Phi$. The interest
in such waves is also motivated by the quest to derive
selection criteria for the na\"ive continuum limit of~\eqref{eq:FPUintro}, which is the PDE
$\partial_{tt}u=\partial_x\Phi^\prime\at{\partial_x}$. For non-convex $\Phi$, this equation is ill-posed
due to its elliptic-hyperbolic nature, and one proposal is to select solutions by so-called kinetic relations~\cite{Abeyaratne:91a,Truskinovski:87a} derived from travelling waves in atomistic models.
\par
Combining the travelling wave ansatz $u_j\at{t}=U\at{j-ct}$ with \eqref{eq:FPUintro} yields
the delay-advance-differential equation
\begin{align}
\label{eq:strain}
c^2 R'' (x) = \Delta_1 \Phi' (R(x)) ,
\end{align}
where $R(x) := U(x+1/2) - U(x-1/2)$ the (symmetrised) discrete strain profile and
$\Delta_1 F(x) := F (x + 1) - 2F(x) + F(x - 1)$. Periodic and homoclinic travelling waves have
been studied intensively, see
\cite{Friesecke:94a,Smets:97a,Friesecke:99a, Pankov:05a,EP05,Her10a} and the references therein,
but very little is known about heteroclinic waves. The authors are only aware of
\cite{HR10b,Herrmann:11b}, which prove the existence of supersonic heteroclinic waves,
and the small amplitude results from \cite{Iooss:00b}. In particular, there seems to be no result
that provides phase transitions waves with large amplitudes for generic double-well potentials.
\par
Phase transition waves with large amplitudes are only well understood
for piecewise quadratic potentials, and there exists a rich body of literature on bi-quadratic potentials, starting with
\cite{Balk:01a,Balk:01b,TV05}.  For the special case
\begin{align}
\label{Intro:Eqn.0}
\Phi\at{r}=\tfrac{1}{2}r^2-\abs{r}\,,\qquad \Phi^\prime\at{r}=r-\sgn\at{r}
\end{align}
the existence of phase transition waves has been established by two of the authors using rigorous Fourier methods. In \cite{Schwetlick:07a} they consider subsonic speeds $c$ sufficiently close to $1$, which is the speed of sound,
and show that \eqref{eq:strain} admits for each $c$ a two-parameter family of waves. These waves have
exactly one interface and connect different periodic tail oscillations, see Figure \ref{Fig:waves} for an illustration.
\par
\begin{mhchange}
In this paper we allow for small perturbations of the potential \eqref{Intro:Eqn.0} and show that the three-parameter family of phase transition waves from \cite{Schwetlick:08a} persist provided that the perturbation is sufficiently localised
with respect to the strain variable $r$.
\par
A related problem has been studied in \cite{Vainchtein:09a}. There, a piecewise quadratic family of potentials is considered such that the stress-strain relationship is continuous and trilinear, with a small spinodal region. Travelling wave solutions are shown to obey a relation of residuals in the Fourier representation, which is then approximately solved numerically. The regularity of the perturbed potential is lower than that of the class of perturbations considered here, so strictly speaking the results do not overlap. However, in spirit the settings are close and indeed the numerical evidence~\cite[Fig.~4, bottom right panel]{Vainchtein:09a} is in good in agreement with our findings: there is an one-sided asymptotically constant solution, and the tail behind the interface oscillates with slightly different amplitude from that related to~\eqref{Intro:Eqn.0}. The range of velocities considered in~\cite{Vainchtein:09a} is larger than the one studied here. 
\par
Our approach is in essence perturbative and reformulates the travelling wave equation with perturbed potential
in terms of a corrector profile $S$, i.e., we write $R=R_0+S$, where $R_0$ is a given wave
corresponding to
the unperturbed potential. The resulting equation for the corrector $S$ can be written as
\begin{align}
\label{Intro:Eqn.1}
\calM {S}=\calA^2 \calG\at{S}+\eta\,,
\end{align}
where $\eta$ is a constant of integration and $\calA$, $\calM$, $\calG$
are operators to be identified below. More precisely, $\calM$ is a linear integral operator which depends on $c$
and $\calG$ a nonlinear superposition operator involving $R_0$. The analysis of \eqref{Intro:Eqn.1} is rather delicate since the Fourier symbol of $\calM$ has real roots, which implies that $0$ is an inner point of
the continuous spectrum of $\calM$.  In particular  $\calM$ is not a Fredholm operator in the function spaces we considered here, so a standard bifurcation analysis from $\delta=0$ via a Lyapunov-Schmidt reduction is impossible. 
\par
In our existence proof,
we first eliminate the corresponding singularities and derive an appropriate solution formula for the
linear subproblem. Afterwards we introduce a class
of admissible functions $S$ and show that the properties of $\calA$ and $\calG$ guarantee that $\calA^2\calG\at{S}$ is compactly supported and sufficiently small. These fine properties are illustrated in Fig.~\ref{Fig:g_fct} and allow us to define
a nonlocal and nonlinear operator $\calT$ such that
\begin{align*}
\calM\calT\at{S}=\calA^2\calG\at{S}+\eta(S)
\end{align*}
holds for all admissible $S$ with some $\eta\at{S}\in\Rset$. This operator
$\calT$ is contractive in some ball of an appropriately defined function space,
so the existence of phase transitions waves is granted by the contraction mapping
principle, see Lemma \ref{Lem:SelfMapping}. Moreover, the properties
of $\calM$ and $\calG$ imply that our fixed point method for $S$ yields all phase transition waves
$R$ that comply with certain requirements, see Proposition \ref{Thm:Uniqueness}.
\end{mhchange}
\par
\begin{mhchange}
Our existence result yields -- for each $c$ from an interval of subsonic velocities -- a genuine two-parameter family of solutions to \eqref{eq:strain} but it is not clear whether all these phase transition waves are physically reasonable.
In the literature, one often employs
selection criteria to single out a unique phase transition wave for each speed $c$. One selection criterion is the \emph{causality principle},
which in our case selects waves with non-oscillatory tails in front of the interface;  see
\cite{Sle01,Sle02,TV05} and Remark \ref{MainResult.CommentCausality} following the Main Theorem~\ref{Thm:Main}.
These waves can also be observed in numerical simulations of atomistic Riemann problems with non-oscillatory initial data~\cite{Herrmann:10b}.
\par
Below we tailor our perturbation method carefully in order to show persistence of the  
tail oscillations in front of the interface. In particular, for each small $\delta$ and any given $c$ 
we obtain exactly one wave that propagates towards an asymptotically constant state. The
other solutions are oscillatory for both $x\to-\infty$ and $x\to+\infty$,
and satisfy the \emph{entropy principle} -- which is less restrictive than the causality principle -- 
as long as the oscillations in front of the interface have smaller amplitude than those behind;
see \cite{Herrmann:10b} for more details and a discussion of the different versions of Sommerfeld's \emph{radiation condition}. 
It is not known 
whether waves with tail oscillations on both sides of the interface are dynamically stable
or can be created by Riemann initial data. As usual, however, one might expect that travelling waves 
resemble dynamical solutions of more complex situations in a temporal and spatial window. Candidates are
cascades of moving phase interfaces in chains with multi-well potential or other types of macroscopically self-similar waves.
\end{mhchange}
\par
We also emphasise that phase transition waves satisfy Rankine-Hugoniot conditions for
the macroscopic averages of mass, momentum, and total energy \cite{Herrmann:10b},
which imply nontrivial restrictions between the wave speed and the tail oscillations on both sides of the interface.
Although these conditions do not appear explicitly in our existence proof, they can (at least in principle) be computed
because the tail oscillations are given by harmonic waves, see
again Fig.~\ref{Fig:waves}. For general double-well potentials, however,
it is much harder to evaluate the Rankine-Hugoniot conditions and thus it remains unclear which tail oscillations can be connected by phase transition waves. Closely related to the jump condition for the total energy is the \emph{kinetic relation}, which specifies the transfer between oscillatory and non-oscillatory energy at the interface and determines
the \emph{configurational force} that drives the wave. In the final section
we discuss how the kinetic relation changes to leading order under small perturbations of the potential \eqref{Intro:Eqn.0}.
%
%
\bigpar
We now present our main result in greater detail.
%

%
%
\subsection{Overview and main result}
%
%
We study an atomic chain with interaction potential
\begin{align*}
\Phi_\delta\at{r}=\tfrac12 r^2-\Psi_\delta\at{r}\,,\qquad \Psi_\delta\at{0}=0\,,
\end{align*}
where $\Psi^\prime_\delta$ is a perturbation of $\Psi_0^\prime=\sgn$ in a small neighbourhood of 0.
The travelling wave equation therefore reads
\begin{align}
\label{Eqn:TW1}
c^2 R^{\prime\prime}=\triangle_1\at{R-\Psi_\delta^\prime\at{R}}
\end{align}
and depends on the parameters $c$ and $\delta$.  In order to show that \eqref{Eqn:TW1} admits solutions
for small $\delta$ we rely on the following assumptions on $\Psi^\prime_\delta$, see  Figure \ref{Fig:potential} for an illustration.
\begin{assumption}
\label{Ass:DPsi}
Let $\at{\Psi_\delta}_{\delta>0}$ be a one-parameter family of $\fspaceC^2$-potentials such that
\begin{enumerate}
\item
$\Psi^\prime_\delta$ coincides with $\Psi^\prime_0$ outside the interval $\oointerval{-\delta}{\delta}$,
\item
there is a constant
$\DPsi$ independent of $\delta$ such that
\begin{align*}
\abs{\Psi_\delta^\prime\at{r}}\leq{\DPsi}\,,\qquad
\abs{\Psi_\delta^{\prime\prime}\at{r}}\leq\frac{\DPsi}{\delta}\,.
\end{align*}
for all $r\in\Rset$.
\end{enumerate}
\end{assumption}
The quantity
\begin{align*}
I_\delta := \tfrac{1}{2}\int_\Rset \bat{\Psi_\delta^\prime\at{r}-\Psi_0^\prime\at{r}}\dint{r}
\end{align*}
plays in important role in our perturbation result as it determines the leading order correction. Notice that our assumptions imply
\begin{align*}
I_\delta = \tfrac{1}{2}\int^\delta_{-\delta}\Psi_\delta^\prime\at{r}\dint{r}=- \tfrac{1}{2}\at{\Phi_\delta\at{+1}-\Phi_\delta\at{-1}}
\qquad\text{and hence}\qquad
\abs{I_\delta}\leq C_\Psi\delta \,.
\end{align*}
As already mentioned, the case $\delta=0$ has been solved in \cite{Schwetlick:07a}. The main result can be summarised as follows.

\begin{proposition}[\cite{Schwetlick:07a}, Proof of Theorem 3.11 and~\cite{Schwetlick:08a}, Theorem 1]
\label{Prop:ExistenceOfAnchor}
There exist $0<c_0 < 1$ such that for every $c \in[c_0,1)$, there exists a two-parameter family of
solutions $R_0\in\fspaceW^{2,\infty}\at\Rset$ to the travelling wave equation \eqref{Eqn:TW1} with $\delta=0$.
 This family is normalised by $R_0\at{0}=0$ and can be described as follows:
\begin{enumerate}
\item
\begin{mhchange} There exists a unique travelling wave $\bar{R}_0$ such that
\begin{align*}
  \bar{R}_0\at{x}\quad & \xrightarrow{\;x\to+\infty\;}\quad \bar r^+_c \\
\bar{R}_0\at{x}-\alpha^-_c\bat{\cos\at{k_cx}-1}-\beta^-_c\sin\at{k_cx}\quad&\xrightarrow{\;x\to-\infty\;}\quad \bar r^-_c
\end{align*}
for some constants $r^\pm_c$, $k_c$, $\alpha^-_c$, and $\beta^-_c$ depending on $c$.
\item There exists an open neighbourhood $U_c$ of $0$ in $\Rset^2$ such that for any $\pair{\alpha}{\beta}\in U_c$
the function $R_0=\bar{R}_0+\alpha\bat{\cos\at{k_c\cdot}-1}+\beta\sin\at{k_c\cdot}$ is a travelling wave with \end{mhchange}
\begin{enumerate}
\item  \label{it:anch-one} $\norm{R_0}_{\infty} \leq D_0\at{1-c^2}^{-1}$\,,
\item $R_0\at{x}>r_0$ for $x>x_0$ and $R_0\at{x}<-r_0$ for $x<-x_0$\,,
\item \label{it:anch-two} $R_0^\prime\at{x}>d_0$  for $\abs x<x_0$
\end{enumerate}
for some constants $x_0$, $r_0$, $d_0$, and $D_0$ depending on $c_0$.
\end{enumerate}
\end{proposition}
\begin{figure}[ht!]%
\centering{%
\includegraphics[width=0.65\textwidth, draft=\figdraft]%
{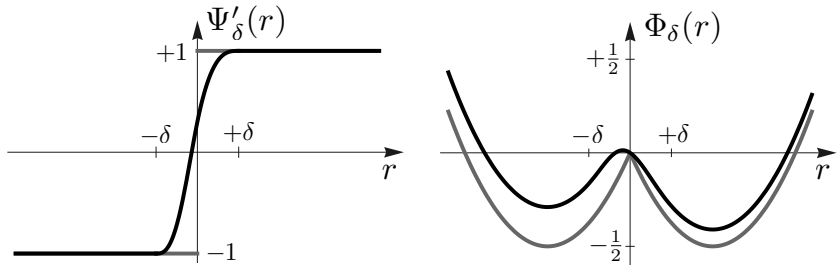}%
}%
\caption{%
Sketch of $\Psi_\delta^\prime$ and $\Phi_\delta$ for $\delta=0$ (grey) and $\delta>0$ (black).
Since $\Phi_0$ is symmetric, $-I_\delta$ is just half the energy difference between the two wells of $\Phi_\delta$.
}%
\label{Fig:potential}%
\end{figure}
\begin{figure}[ht!]%
\centering{%
\includegraphics[width=0.5\textwidth, draft=\figdraft]%
{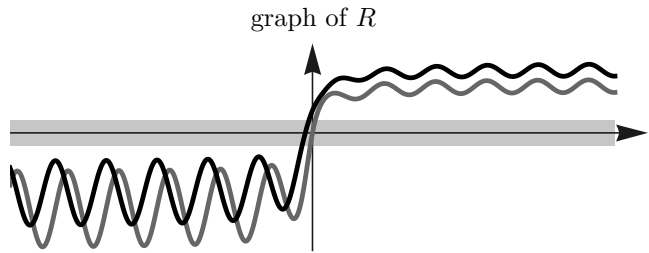}%
}%
\caption{%
Sketch of the waves for $\delta=0$ (grey) and $\delta>0$ (black) as provided by our perturbation result;
the shaded region indicates the spinodal interval $\ccinterval{-\delta}{+\delta}$, where  $\Psi^\prime_\delta$ differs from $\Psi^\prime_0$. Both waves differ by the constant $I_\delta=\DO{\delta}$ and a small corrector $S$
of order $\DO{\delta^2}$, which is oscillatory for $x<0$ but asymptotically constant as $x\to+\infty$. The tail oscillations of both waves do not penetrate
the spinodal region and are generated by harmonic waves with wave number $k_c$.
For each admissible $\delta$ and $c$ there exists exactly on wave
that satisfies the causality principle as it is non-oscillatory for $x\to+\infty$. 
}%
\label{Fig:waves}%
\end{figure}
\bigpar
The main result of this article can be described as follows.
\begin{theorem}
\label{Thm:Main}
For all $c_1\in\oointerval{c_0}{1}$
there exists $\delta_0>0$ such that for any $0<\delta<\delta_0$, any speed $c_0<c<c_1$, and any given wave
$R_0$ as in Proposition \ref{Prop:ExistenceOfAnchor} there exists a solution $R$  to \eqref{Eqn:TW1} with
\begin{align*}
R=R_0-I_\delta +S,
\end{align*}
where $I_\delta = \DO{\delta}$ and the corrector $S\in\fspace W^{2,\infty}\at\Rset$
\begin{enumerate}
\item
vanishes at $x=0$,
\begin{mhchange}
\item
is non-oscillatory as $x\to+\infty$, i.e., the limit $\lim_{x\to+\infty}S\at{x}$ is well defined,
\item
admits harmonic tail oscillations for $x\to-\infty$, that means there exists constants
$a_-$ and $d_-$ such that  $\lim_{x\to-\infty}S\at{x}-a_-R_0\at{x+d_-}$ is well defined,
\end{mhchange}
\item
is small in the sense of
\begin{align*}
\norm{S}_{\infty}=\nDO{\delta^2},\qquad
\norm{S^\prime}_{\infty}=\DO{\delta},\qquad
\norm{S^{\prime\prime}}_{\infty}=\DO{1}.
\end{align*}
\end{enumerate}
Moreover, the solution $R$ with these properties is unique provided that $\delta$ is sufficiently small.
\end{theorem}
More detailed information about the existence and uniqueness part of our result are given in
Proposition~\ref{Thm:Existence} and  Proposition~\ref{Thm:Uniqueness}, respectively. We further mention:
\begin{enumerate}
\item 
Since the travelling wave equation is invariant under
\begin{align*}
c\rightsquigarrow-c,\qquad R\at{x}\rightsquigarrow R\at{-x},
\end{align*}
there exists an analogous result for $-1<c\ll0$. 
\item
Different choices of $c$ and $R_0$ provide different waves $R$, see Section \ref{sect:uniqueness}.
\item
The travelling wave equation \eqref{Eqn:TW1} is, of course,
invariant under shifts in $x$ but fixing $R_0$ and $S$ at $0$ removes neutral directions in the contraction proof.
\item
All constants derived below depend on $c_1$ and $c_0$ but for notational simplicity we do not write this dependence explicitly.
It remains open whether $\delta_0$ can be chosen independently of $c_1$.
\item\label{MainResult.CommentCausality} The \emph{causality principle} selects those solutions   with $c_\mathrm{gr}<c_\mathrm{ph}$ and $c_\mathrm{gr}>c_\mathrm{ph}$   for all oscillatory harmonic modes ahead and behind the interface, respectively, where, $c_\mathrm{gr}$ and $c_\mathrm{ph}$ are the group and the phase velocity. For nearest neighbour 
chains with interaction potential $\Phi_0$ and wave speed $c$ sufficiently close to $1$,  Proposition \ref{Prop:ExistenceOfAnchor} yields
\begin{align*}
c_{\mathrm{ph}}=c=a\at{k_c}=k_c^{-1}\Om\at{k_c}>
c_{\mathrm{gr}} =\Om^\prime\at{k_c}
\end{align*}
on \emph{both} sides of the interface, where $\Om\at{k}=2\sin\at{k/2}$ is the dispersion relation \cite{SCC05,Herrmann:10b}.  The causality principle therefore selects the solution $\bar{R}_0$
as it is the only wave having no tail oscillations ahead of the interface. Since our perturbative approach
changes neither the wave speed $c$ nor the wave 
number $k_c$ in the oscillatory modes (but only the amplitude behind the interface and, of course, the behaviour near the interface), we conclude that Theorem~\ref{Thm:Main} provides for each $\delta$ and $c$ exactly one wave that complies with the causality principle.
\item
The surprisingly simple leading order effect, that is the addition of $-I_\delta$ to $R_0$, implies that the kinetic relation does not change to order $\nDO\delta$.  Notice, however, that the kinetic relation depends on the choice of $R_0$, cf. \cite{Schwetlick:08a}.
\end{enumerate}
\bigpar
\begin{mhchange}
This paper is organised as as follows. In Section~\ref{sect:Prelim} we reformulate \eqref{Eqn:TW1} in terms of integral operators $\calA$ and $\calM$ and show that it is sufficient to prove the existence of waves for the
special case  $I_\delta=0$. Section~\ref{sect:proof} concerns
the existence of correctors $S$. We first establish an inversion formula for $\calM$ which in turn enables us to
define an appropriate solution operator $\calL$ to the affine subproblem $\calM S = \calA^2G+\eta$ with given $G$.
Afterwards we investigate the properties of the nonlinear operator $\calG$ and
prove the contractivity of the fixed point operator $\calT$. In Section~\ref{sect:uniqueness}
we establish our uniqueness result and conclude with a discussion of the kinetic relation in Section~\ref{sect:kinrel}.
\end{mhchange}
%
%
\section{Preliminaries and reformulation of the problem}\label{sect:Prelim}
%
%
%
In this section we reformulate the travelling wave equation \eqref{Eqn:TW1} in terms of integral
operators and show that elementary transformations allow us to assume that
$I_\delta=0$ holds for all $\delta>0$.
%
%
\subsection{Reformulation as integral equation}
%
%
For our analysis it is convenient to reformulate the problem in terms of the convolution operator $\calA$
and the operator $\calM$ defined by
\begin{align*}
\at{\calA{F}}\at{x}:=\int\limits_{{x}-1/2}^{{x}+1/2}F\at{s}\dint{s}\,,\qquad
\calM F:=\calA^2 F-c^2F\,.
\end{align*}
The travelling wave equation can then be stated as
\begin{align}
\label{Eqn:TW2}%
\calM{R}=\calA^2\Psi^\prime_\delta\at{R} + \mu\,.
\end{align}
\begin{lemma}
\label{Lem:IntEqn}
A function $R\in\fspaceW^{2,\infty}\at\Rset$ solves the travelling wave equation \eqref{Eqn:TW1} if and only if
there exists a constant $\mu\in\Rset$ such that $\pair{R}{\mu}$ solves \eqref{Eqn:TW2}.
\end{lemma}
\begin{proof}
By definition of $\calA$, we have $\frac{\dint^2}{\dint{x}^2}\calA^2=\triangle_1$.
Equation \eqref{Eqn:TW1} is therefore, and due to the definition of $\calM$, equivalent to
\begin{align}
\label{Lem:IntEqn.Eqn1}
\at{\calM {R}}^{\prime\prime}=P^{\prime\prime}\,,\qquad
P:=\calA^2\Psi_\delta^\prime\at{R}\,.
\end{align}
The implication \eqref{Eqn:TW2}$\implies$\eqref{Eqn:TW1} now follows immediately. Towards the
reversed statement, we integrate \eqref{Lem:IntEqn.Eqn1}$_{1}$ twice with respect to $x$ and
obtain $\calM R=P+\la x + \mu$, where $\la$ and $\mu$ denote constants of integration. The condition
$R\in\fspaceL^\infty\at\Rset$ implies  $\calM{R},\,\Psi^\prime_\delta\at{R},\,P\in\fspaceL^\infty\at\Rset$, and we
conclude that $\la=0$.
\end{proof}
%
%
\subsection{Properties of the operators \texorpdfstring{$\calA$ and $\calM$}{...}}
%
Some of our arguments rely on Fourier transform, which we
normalise as follows
\begin{align*}
\widehat{F}\at{k}= \frac{1}{\sqrt{2\pi}}\int_\Rset\mhexp{\iu kx}F\at{x}\dint{x}\,,\qquad
F\at{x} = \frac{1}{\sqrt{2\pi}}\int_\Rset\mhexp{-\iu kx}\widehat{F}\at{k}\dint{k}\,.
\end{align*}
\begin{figure}[ht!]%
\centering{%
\includegraphics[width=0.3\textwidth, draft=\figdraft]%
{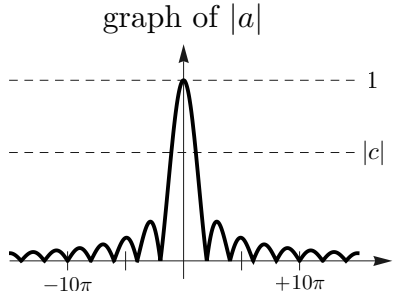}%
}%
\caption{%
The the real roots of the symbol function $m$ are the solutions to $\abs{a\at{k}}=\abs{c}$.
}%
\label{Fig:symbol}%
\end{figure}%
Using standard techniques for the Fourier transform in the space of tempered distributions we readily
verify the following assertions.
\begin{remark}
\label{Rem:MProps}
The operators $\calA$ and $\calM$ diagonalise in Fourier space and
have symbols
\begin{align*}
a\at{k}=\frac{\sin\at{k/2}}{k/2}
\quad\text{ and }\quad
m\at{k}=a\at{k}^2-c^2\,,
\end{align*}
respectively. \begin{mhchange} In particular, we have
\begin{align*}
\calM\cos\at{k_c\cdot}=0\,,\qquad
\calM\sin\at{k_c\cdot}=0\,,\qquad \calM{1}=1-c^2\,,
\end{align*}
for any real root $k_c$ of $m$, and
\begin{align*}
F\in\xspan\big\{\cos\at{k_c\cdot},\,\sin\at{k_c\cdot}\;:\; m\at{k_c}=0,\quad k_c>0\big\}
\end{align*}
for any tempered distribution $F$ with $\calM{F}=0$.
\end{mhchange}
\end{remark}
\begin{mhchange}
The set of real roots of $m$ depends strongly on the value of $c$, see Figure \ref{Fig:symbol}. In what follows
we only deal with positive and near sonic speed $c$, that means $c\lessapprox1$, for which $m$ has two simple real roots.
\end{mhchange}
\bigpar
We next summarise further properties of the operator $\calA$ and recall that the Sobolev space $\fspaceW^{1,\,p}\at\Rset$ is for any $1\leq{p}\leq\infty$ continuously embedded into $\fspace{BC}\at\Rset$. 
\begin{lemma}
\label{Lem:A.Props}
For any $1\leq{p}\leq\infty$ we have $\calA\colon\fspaceL^p\at\Rset\to{\fspace{W}^{1,p}\at\Rset}\subset
\fspace{BC}\at\Rset$ with
\begin{align}
\label{Lem:A.Props.Eqn1}
\norm{\calA{F}}_{p}\leq\norm{F}_{p}\,,\qquad
\norm{\at{\calA{F}}^\prime}_{p}\leq 2\norm{F}_{p}\,,\qquad
\norm{\calA{F}}_{\infty}\leq\norm{F}_{p}
\end{align}
for all $F\in\fspaceL^p\at\Rset$, where $\at{\calA F}^\prime=\nabla{F}:=F\at{\cdot+\tfrac{1}{2}}-
F\at{\cdot-\tfrac{1}{2}}$.  Moreover, $\supp F\subseteq\ccinterval{x_1}{x_2}$ implies
$\supp{\calA{F}}\subseteq\ccinterval{x_1-\tfrac{1}{2}}{x_2+\tfrac{1}{2}}$.
\end{lemma}
\begin{proof}
Let $1\leq{p}<\infty$ and $F\in\fspaceL^p\at\Rset$.
The definition of $\calA$ ensures
that $\calA{F}$ has in fact the weak derivative $\nabla{F}$, and this implies
the estimate \eqref{Lem:A.Props.Eqn1}$_2$ via
$\norm{\nabla F}_{p}\leq 2\norm{F}_{p}$.
Using H\"older's inequality we find
\begin{align*}
\babs{\at{\calA F}\at{x}}^p\leq \int_{x-1/2}^{x+1/2}\babs{F\at{s}}^p\dint{s}
\end{align*}
and integration with respect to $x$ yields \eqref{Lem:A.Props.Eqn1}$_1$. We also infer that
$\babs{\at{\calA F}\at{x}}\leq \norm{F}_{p}$ holds for all $x\in\Rset$, and this gives
\eqref{Lem:A.Props.Eqn1}$_3$. Finally,
the arguments for $p=\infty$ are similar and the claimed relation between $\supp{F}$ and $\supp{\calA{F}}$ is a direct consequence of the definition of $\calA$.
\end{proof}
%
%
\subsection{Transformation to the special case \texorpdfstring{$I_\delta=0$}{...}}
%
%
The key observation that traces the general case $I_\delta\neq0$ back to the special case
$I_\delta=0$ is that
any shift in $\Psi^\prime_\delta$ can be compensated by adding a constant to $R$.
\begin{lemma}
\label{Lem:Trafo}
The family $\nat{\tilde{\Psi}_{\tilde{\delta}}}_{\tilde\delta>0}$ defined by
\begin{align*}
\tilde{\delta}=\delta\at{1+C_\Psi}, \qquad
\tilde{\Psi}^{\prime}_{\tilde{\delta}}\at{r} = \Psi^\prime_\delta\at{r-I_\delta }
\end{align*}
satisfies Assumption \ref{Ass:DPsi} with constant $\tilde{C}_\Psi=C_\Psi\at{1+C_\Psi}$ as well as
\begin{align*}
\tilde{I}_{\tilde{\delta}}=\tfrac12\int_\Rset
\tilde{\Psi}^{\prime}_{\tilde{\delta}}\at{r} - \Psi^{\prime}_{0}\at{r}\dint{r}=0\quad\text{ for all}\quad \tilde\delta>0.
\end{align*}
Moreover, each solution $\npair{\tilde{R}}{\tilde\mu}$ to the modified travelling wave equation
\begin{align}
\label{Eqn:TW2a}
\calM \tilde{R}=\calA^2\tilde{\Psi}_{\tilde\delta}^\prime\nat{\tilde{R}}+\tilde{\mu}
\end{align}
defines a solution $\pair{R}{\mu}$ to \eqref{Eqn:TW2} via $R=\tilde{R}-I_\delta$ and $\mu=\tilde\mu-\at{c^2-1}I_\delta$ and vice versa.
\end{lemma}
\begin{proof}
Due to $\abs{I_\delta}\leq C_\Psi\delta$ and our definitions we find
$
\tilde{\Psi}^\prime_{\tilde\delta}\at{r}=\Psi_0^\prime\at{r}$
at least for all $r$ with $\abs{r}\geq \tilde{\delta}$, as well as
\begin{align*}
\abs{\tilde{\Psi}_{\tilde\delta}^\prime\at{r}}\leq{C_\Psi}\leq\tilde{C}_\Psi\,,\qquad
\abs{\tilde{\Psi}_{\tilde\delta}^{\prime\prime}\at{r}}\leq \frac{{C_\Psi}}{\delta}= \frac{{C_\Psi}}{\delta} \frac{1+C_\Psi}{1+C_\Psi}=\frac{ \tilde{C}_\Psi}{\tilde\delta}\quad \text{for all}\quad r\in\Rset\,.
\end{align*}
We also have
\begin{align*}
\tilde{I}_{\tilde\delta} &=  \tfrac{1}{2}\int_\Rset \bat{\Psi_\delta^\prime\at{r-I_\delta}-\Psi_0^\prime\at{r}}\dint{r}= \tfrac{1}{2}\int_\Rset \bat{ \Psi_\delta^\prime\at{r}-\Psi_0^\prime\at{r+I_\delta}}\dint{r}\\
&=\tfrac{1}{2}\int_\Rset \bat{ \Psi_\delta^\prime\at{r}-\Psi_0^\prime\at{r}}\dint{r}+\tfrac{1}{2}\int_\Rset \bat{ \Psi_0^\prime\at{r}-\Psi_0^\prime\at{r+I_\delta}}\dint{r}=I_\delta -I_\delta=0\,.
\end{align*}
Finally, the equivalence of \eqref{Eqn:TW2} and \eqref{Eqn:TW2a} is obvious.
\end{proof}
%
%
\section{Existence of phase transition waves}\label{sect:proof}
%
%
\begin{mhchange}
In this section, we show that each phase transition wave for $\Psi_0$ persists under the perturbation $\Psi_0\rightsquigarrow\Psi_\delta$, provided that $\delta$ is sufficiently small. To this end we proceed as follows.
\end{mhchange}
\begin{enumerate}
\item
We fix $c\in\ccinterval{c_0}{c_1}$ with  $0<c_0<c_1<1$ as in
Proposition \ref{Prop:ExistenceOfAnchor} and Theorem \ref{Thm:Main}. Then
there exists a unique solution $k_c>0$ to $a\at{k_c}=c$, and this implies $m\at{\pm k_c}=0$,
$m^\prime\at{\pm k_c}\neq0$, and $m\at{k}\neq0$ for $k\neq\pm k_c$.
All constants derived below can be chosen independently of $c$ but are allowed to
depend on $c_0$ and $c_1$.
\item
Thanks to  Proposition \ref{Prop:ExistenceOfAnchor} and Lemma \ref{Lem:IntEqn},
we fix $\pair{R_0}{\mu_0}$ from the two-parameter family of solutions to the integrated travelling wave equation \eqref{Eqn:TW2} for $\delta=0$ and given $c$. Recall that   $R_0$ is normalised by $R_0\at{0}=0$.
\item
In view of Lemma \ref{Lem:Trafo}, we assume that $I_\delta=0$ holds for all $\delta>0$. To avoid
unnecessary technicalities we also assume from now on that $\delta$ is sufficiently small.
\end{enumerate}
In order to find a solution $\pair{R}{\mu}$ to the integrated travelling wave equation \eqref{Eqn:TW2} for  $\delta>0$, we further make the ansatz
\begin{align*}
R=R_0+S\,,\qquad\mu=\mu_0+\eta\,,
\end{align*}
and seek correctors $\pair{S}{\eta}$ such that
\begin{align}
\label{FPP:CorrEqn}
\calM S = \calA^2G +\eta\,,\qquad G=\calG\at{S}\,.
\end{align}
Here, the nonlinear operator $\calG$ is defined by
\begin{align}
\label{Eqn:Def.G}
{\calG\at{S}}\at{x}=
\Psi^\prime_\delta\bat{R_0\at{x}+S\at{x}}-\Psi^\prime_0\bat{R_0\at{x}}\,.
\end{align}
\begin{mhchange}
In order to identify a natural ansatz space $\fspace{X}$ for $S$, we first
remark that the smoothing properties of $\calA$, see Lemma \ref{Lem:A.Props}, imply
$S\in\fspaceW^{2,\infty}\at\Rset$. Notice, however, that $R=R_0+S$ is in general more regular due to the smoothness of $\Psi_\delta$. More precisely, \eqref{Eqn:TW2} combined with $\Psi_\delta\in\fspaceC^{k}\at\Rset$ yields $R\in\fspaceC^{k+1}\at\Rset$.  We also impose the normalisation condition $S\at{0}=0$
in order to eliminate the non-uniqueness that results from the shift invariance of
the travelling wave equation \eqref{Eqn:TW2}.  In fact, without this constraint
any corrector $S$ provides a whole family of other possible correctors via
$\tilde{S}=S\at{\cdot+x_0} + R_0\at{\cdot+x_0}-R_0$ with $x_0=\nDO{\delta^2}$.
\par
A key property of our existence and uniqueness result is that
$R$ has only harmonic tail oscillations
with wave number $k_c$ and that both $R$ and $R_0$ share
the same tail oscillations for $x\to+\infty$. The corrector $S$ is therefore non-oscillatory
in the sense that $S\at{x}$ converges as $x\to+\infty$ to some well-defined limit $\sigma$.
In summary, we seek solutions $\pair{S}{\eta}$ to \eqref{FPP:CorrEqn} with
$S\in\fspace{X}$ and $\eta\in\Rset$, where \end{mhchange}
\begin{align*}
\fspace{X}:=\Big\{S\in\fspaceW^{2,\infty}\at\Rset\;:\; S\at{0}=0\,,\quad \begin{mhchange}
\sigma=\lim_{x\to+\infty}S\at{x}\text{ exists}\,\Big\} \end{mhchange}
\end{align*}
\begin{mhchange}
is a closed subspace of $\fspaceW^{2,\infty}\at\Rset$ and hence a Banach space.
\end{mhchange}
%
%
%
%
\subsection{Inversion formula for \texorpdfstring{$\calM$}{the operator M}}
%
\begin{mhchange}
Our first task is to construct for given $G$ a solution $\pair{S}{\eta}$ to the affine equation \eqref{FPP:CorrEqn}$_1$.
In  a preparatory step, we therefore study the solvability of the equation
\begin{align}
\label{Eqn:AuxProblem}
\calM{F}=Q
\end{align}
using the Fourier transform for tempered distributions, where $Q\in\fspaceL^\infty\at\Rset$ is some given function.
This problem is not trivial because the symbol function $m$ has two simple roots at $\pm k_c$, or, equivalently,
because $0$ is an element of the continuous spectrum of $\calM$ corresponding to a two-dimensional space of
generalised eigenfunctions. We are therefore confronted with the following two issues in
Fourier space:
\begin{enumerate}
\item
$\widehat{F}$ is uniquely determined only up to
elements from the space
\begin{align*}
\xspan \big\{\,\delta_{-k_c}\at{k}\,,\,\delta_{+k_c}\at{k}\,\big\},
\end{align*}
which contains the Fourier transforms of all
bounded kernel functions of $\calM$.
\item
$\widehat{F}$ exhibits -- at least for generic $Q$ with $\widehat{Q}\at{\pm k_c}\neq0$ -- two poles at $\pm k_c$
and is hence not Lebesgue integrable in the vicinity of $\pm k_c$. In particular,
the dual pairing between $\widehat{F}$ and a Schwartz function is defined in the sense of Cauchy principal values only.
\end{enumerate}
The non-uniqueness is actually an advantage because it allows us to select
solutions with particular properties; see the proof of Lemma \ref{Lem:LinSolOp},
where we add an appropriately chosen kernel function to ensure
non-oscillatory behaviour for $x\to+\infty$. Concerning the non-integrable
poles at $\pm k_c$, we split $\widehat{F}$ into a two-dimensional
singular part and a remaining regular part,
and show that any solution $F$ to \eqref{Eqn:AuxProblem} belongs to some Lebesgue space provided that
$\widehat{Q}$ is sufficiently regular.
\end{mhchange}
\begin{figure}[ht!]%
\centering{%
\includegraphics[width=0.9\textwidth, draft=\figdraft]%
{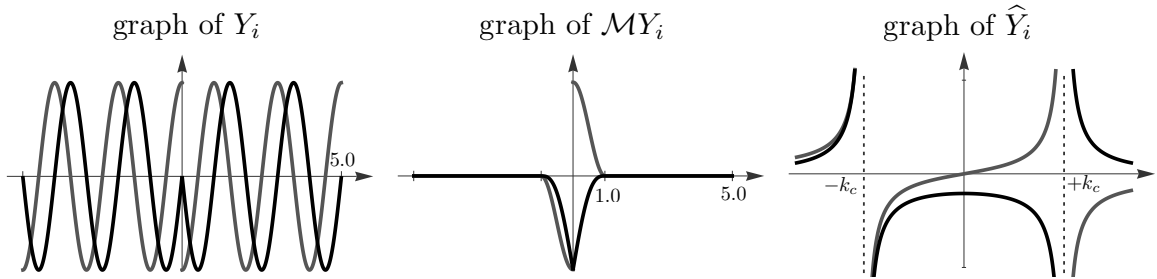}%
}%
\caption{%
Properties of $Y_1$ (grey) and $Y_2$ (black).
}%
\label{Fig:y_fct}%
\end{figure}
\bigpar
As illustrated in  Figure \ref{Fig:y_fct}, we introduce
two functions $Y_1,Y_2\in\fspaceL^\infty\at\Rset$ with
\begin{align*}
Y_{1}\at{x}:=\frac{ \sqrt{2\pi}}{m^\prime\at{k_c}}\cos\at{k_cx}\sgn\at{x}\,,\qquad
Y_{2}\at{x}:=\frac{ \sqrt{2\pi}}{m^\prime\at{k_c}}\sin\at{k_cx}\sgn\at{x}\,,
\end{align*}
and verify by direct computations the following assertions.

\begin{remark}\label{Rem:Y.Props}
We have
\begin{enumerate}
\item
$\calM Y_i\in\fspace{L}^\infty\at\Rset$ \; with \; $\supp\calM Y_i\subseteq\ccinterval{-1}{1}$\,,
\item
$\displaystyle \widehat{Y}_1\at{k}=+\frac{2\iu}{m^\prime\at{k_c}}\frac{k}{k^2-k_c^2}$ \; and \;
$\displaystyle \widehat{Y}_2\at{k}=-\frac{2}{m^\prime\at{k_c}}\frac{ k_c }{k^2-k_c^2}$\,,
\item
$m\widehat{Y}_i\in\fspace{L}^2\at\Rset\cap\fspace{BC}^1\at\Rset$\; with \;
 $\displaystyle\lim_{k\to\pm{k_c}}m\at{k}\widehat{Y}_{1}\at{k}=\pm\iu$ \; and \;
$\displaystyle\lim_{k\to\pm{k_c}}m\at{k}\widehat{Y}_{2}\at{k}=-1$.
\end{enumerate}
\end{remark}
\begin{mhchange}
In particular, $\widehat{Y}_1$ and $\widehat{Y}_2$ have normalised poles at $\pm{k_c}$, and this
allows us to derive the following linear and continuous inversion formula for $\calM$.
\end{mhchange}
\begin{lemma}
\label{Lem:M.Inv}%
Let $Q$ be given with $\widehat{Q}\in\fspaceL^2\at\Rset\cap\fspace{BC}^1\at{\Rset}$. Then there exists
a unique
 $Z\in\fspaceL^2\at\Rset$ such that 
\begin{align}
\label{Lem:M.Inv.Eqn1}
\calM\at{Z -\iu\frac{\widehat{Q}\at{+k_c}-\widehat{Q}\at{-k_c}}{2}Y_1-\frac{\widehat{Q}\at{+k_c}+\widehat{Q}\at{-k_c}}{2}Y_2}={Q}\,,
\end{align}
Moreover, $Z$ depends linearly on $Q$ and satisfies
\begin{align}
\notag
\norm{Z}_{2}\leq{C}\at{
\norm{\widehat{Q}}_{2}+
\norm{\widehat{Q}}_{1,\,\infty}}\,
\end{align}
for some constant $C$ independent of $Q$.
\end{lemma}
\begin{proof}
 The function $\hat{Z}$ with 
\begin{align}
\label{Lem:M.Inv.Eqn5}
\widehat{Z}\at{k}:=\frac{\displaystyle\widehat{Q}\at{k}+\iu\frac{\widehat{Q}\at{+k_c}-\widehat{Q}\at{-k_c}}{2}m\at{k}\widehat{Y}_1\at{k}+
\frac{\widehat{Q}\at{+k_c}+\widehat{Q}\at{-k_c}}{2}m\at{k}\widehat{Y}_2\at{k}}{m\at{k}}
\end{align}
is well-defined and continuously differentiable for $k\neq\pm k_c$. In view of Remark \ref{Rem:Y.Props}, l'H\^ospital's rule ensures that  the limits $\lim_{k\to-k_c}\widehat{Z}\at{k}$ and $\lim_{k\to + k_c}\widehat{Z}\at{k}$ do exist,
and combining this with the integrability properties of $m$ and $\widehat{Q}$ we find
$\widehat{Z}\in\fspaceL^2\at\Rset$. The inverse Fourier transform
${Z}\in\fspaceL^2\at\Rset$ is therefore well-defined by Parseval's theorem, depends linearly on $Q$, and satisfies \eqref{Lem:M.Inv.Eqn1} by construction. With $J:=\ccinterval{-2k_c}{+2k_c}$ we readily verify the estimates
\begin{align*}
\norm{\widehat{Z}}_{\fspaceL^2\at{\Rset\setminus{J}}}&\leq\norm{{m}^{-1}}_{\fspaceL^\infty\at{\Rset\setminus{J}}}
\norm{\widehat{Q}}_{\fspaceL^2\at{\Rset\setminus{J}}}+\Bat{\babs{\widehat{Q}\at{+k_c}}+
\babs{\widehat{Q}\at{-k_c}}}
\Bat{\norm{\widehat{Y}_1}_{\fspaceL^2\at{\Rset\setminus{J}}}+
\norm{\widehat{Y}_2}_{\fspaceL^2\at{\Rset\setminus{J}}}}\\
&\leq{C}\bat{\norm{\widehat{Q}}_{\fspaceL^2\at{\Rset}}+\norm{\widehat{Q}}_{\fspaceL^\infty\at{\Rset}}},
\end{align*}
and Taylor expanding both the numerator and the denominator of the right hand side in \eqref{Lem:M.Inv.Eqn5}
at $k=\pm{k_c}$ we get
\begin{align*}
\norm{\widehat{Z}}_{\fspaceL^2\at{J}}\leq
{C}\norm{\widehat{Z}}_{\fspace{C}^0\at{J}}\leq
{C}\norm{\widehat{Q}}_{\fspace{C}^{1}\at{J}}\,.
\end{align*}
The desired estimate for $\norm{Z}_2$ now follows from $\norm{{Z}}_{\fspaceL^2\at{\Rset}}^2=
\norm{\widehat{Z}}_{\fspaceL^2\at{\Rset\setminus{J}}}^2+\norm{\widehat{Z}}_{\fspaceL^2\at{J}}^2$. Finally,
$Z$ is the unique solution in $\fspaceL^2\at\Rset$ since any other solution
to \eqref{Lem:M.Inv.Eqn1} differs from $Z$ by a linear combination of $\cos\at{k_c\cdot}$ and $\sin\at{k_c\cdot}$,
see Remark \ref{Rem:MProps}.
\end{proof}
\begin{mhchange}%
Lemma \ref{Lem:M.Inv} implies that the linear operator $\calM$ admits a linear and continuous inverse
\begin{align*}
\calM^{-1}\colon \calF^{-1}\bat{\fspaceL^2\at\Rset\cap\fspace{BC}^1\at\Rset}\to\fspaceL^2\at\Rset\oplus\xspan\big\{Y_1,\,Y_2\big\}\,,
\end{align*}
where $\calF^{-1}$ means inverse Fourier transform. The proof of Lemma \ref{Lem:M.Inv} also reveals that
$\calM^{-1}$ can be extended to a larger space since one only needs that
$\widehat{Q}$ is continuously differentiable in some neighbourhood of $\pm k_c$. For our purpose, however, it is sufficient
to assume that $\widehat{Q}\in\fspace{BC}^1\at\Rset$.
We also mention that the constant $C$ in Lemma \ref{Lem:M.Inv}, which is the
Lipschitz constant of $\calM^{-1}$,
is uniform in $c_0<c<c_1$ but will grow with $c_1 \to 1$, due to the definition of $Y_1$ and $Y_2$ and the properties of $m$.
\end{mhchange}
%
%
\subsection{Solution operator to the affine subproblem}
%
%
We are now able to prove that the affine problem \eqref{FPP:CorrEqn}$_1$ admits a solution operator
\begin{mhchange}
\begin{align*}
\calL \colon G\in\fspace{Y}\mapsto \pair{S}{\eta} \in\fspace{X}\times\Rset\,,
\end{align*}
where
\begin{align*}
\fspace{Y}:=\Big\{G\in\fspaceL^\infty\at\Rset\;:\;\supp G\subseteq\ccinterval{-1}{1}\Big\}\,.
\end{align*}
The existence of $\calL$ is a consequence of the
following result.
\end{mhchange}
\begin{lemma}
\label{Lem:LinSolOp}
\begin{mhchange}
For each $G\in\fspace{Y}$ there exists a unique $\pair{S}{\eta}\in\fspace{X}\times\Rset$
such that
\end{mhchange}
\begin{align}
\label{Lem:LinSolOp.Eqn0}
\calM{S}=\calA^2 G +\eta\,.
\end{align}
\begin{mhchange}
Moreover, $S$ and $\eta$ depend linearly on $G$ and we have
\end{mhchange}
\begin{enumerate}
\item
$\abs{\eta}\leq{C_{\calM}}\norm{\calA^2G}_{\infty}\,,$
\item
$\norm{S}_{\infty}\leq{C_{\calM}}\norm{\calA^2 G}_{\infty}\,,$
\item
$\norm{S^\prime}_{\infty}\leq{C_{\calM}}\norm{\calA G}_{\infty}\,,$
\item
\label{it:reg} $\norm{S^{\prime\prime}}_{\infty}\leq{C_{\calM}}\norm{G}_{\infty}$.
\end{enumerate}
for some constant $C_{\calM}>0$  independent of $G$.
\end{lemma}%
\begin{proof}
The function $Q:=\calA^2 G$ satisfies $\supp Q\subseteq\ccinterval{-2}{2}$, and using
\begin{align*}
 \abs{\widehat{Q}\at{k}}+ \abs{\frac{d}{dk}\widehat{Q}\at{k}}\leq{C}\int\limits_{-2}^2  \bat{1+\abs{x}}\abs{Q\at{x}} \dint{x}\leq C
\norm{Q}_{\infty}\,
\end{align*}
as well as $\norm{\widehat{Q}}_{2}=\norm{Q}_{2}$, we easily verify that
\begin{align*}
\norm{\widehat{Q}}_{2}+\norm{\widehat{Q}}_{1,\,\infty}\leq
C\norm{Q}_{\infty}\,.
\end{align*}
\begin{mhchange}
By Lemma \ref{Lem:M.Inv}, the function $\tilde{S}:=\calM^{-1}\calA^2{G}$ takes the form
$\tilde{S}=Z+f_1Y_1+f_2Y_2$, where
$Z\in\fspaceL^2\at\Rset$ and $f_1,f_2\in\Rset$ satisfy
\begin{align}
\label{Lem:LinSolOp.Eqn1}
\norm{Z}_{2}+\abs{f_1}+\abs{f_2}\leq C\norm{Q}_{\infty}={C}\norm{\calA^2 G}_{\infty}\,.
\end{align}
In particular, we have $\calM\tilde{S}=\calA^2{G}$ and hence
\begin{align*}
c^2 Z = \calA^2Z-\calA^2G + f_1\calM Y_1 + f_2\calM Y_2\,.
\end{align*}
The functions $\calM Y_1$, $\calM Y_2$ are supported in $\ccinterval{-1}{+1}$,
see Remark \ref{Rem:Y.Props}, and ${G}\in\fspace{Y}$ combined with Lemma \ref{Lem:A.Props} implies that
$\calA^2G$ vanishes outside of $\ccinterval{-2}{+2}$. For $\abs{x}\geq2$
we therefore find
\begin{align*}
c^2\babs{Z\at{x}}=\babs{\at{\calA^2Z}\at{z}}\leq
\at{\int_{x-1/2}^{x+1/2}\Bat{\at{\calA Z}\at{s}}^2\dint{s}}^{1/2}
\,\xrightarrow{\;\;x\to\pm\infty\;\;}\;0\,,
\end{align*}
thanks to H\"older's inequality and since Lemma \ref{Lem:A.Props} implies $\calA Z\in\fspaceL^2\at\Rset$. By definition of $\calM$, $Q$, and $\tilde{S}$ we also have
\begin{align}
\label{Lem:LinSolOp.Eqn2}
c^2\tilde{S}=-\calA^2G+\calA^2 \tilde{S}=-\calA^2G+\calA^2\bat{Z+f_1Y_1+f_2Y_2}\,,
\end{align}
and Lemma \ref{Lem:A.Props} ensures that
\end{mhchange}%
\begin{align*}
\norm{\calA^2 Z}_{\infty}\leq\norm{Z}_{2}\,,\qquad
\norm{\calA^2 Y_i}_{\infty}\leq\norm{Y_i}_{\infty}\,.
\end{align*}
Combining these estimates  with  \eqref{Lem:LinSolOp.Eqn1} and \eqref{Lem:LinSolOp.Eqn2}, we arrive at
$\tilde{S}\in\fspaceL^\infty\at\Rset$ with
\begin{align*}
\norm{\tilde{S}}_{\infty}\leq C \norm{\calA^2G}_{\infty}\,.
\end{align*}
Moreover, differentiating the first identity in \eqref{Lem:LinSolOp.Eqn2} with respect to $x$, we get
\begin{align*}
c^2\tilde{S}^\prime=\nabla\bat{-\calA G+\calA \tilde{S}}\,,\qquad
c^2\tilde{S}^{\prime\prime}&=\nabla\nabla\bat{-G+\tilde{S}}\,,
\end{align*}
where the discrete differential operator $\nabla$ is defined as
$\nabla{U}=U\at{\cdot+\tfrac{1}{2}}-
U\at{\cdot-\tfrac{1}{2}}$, cf.\ Lemma \ref{Lem:A.Props}.
This implies
\begin{align*}
\norm{\tilde{\calS}^\prime}_{\infty}\leq{C}\norm{\calA G}_{\infty}\,,\qquad
\norm{\tilde{\calS}^{\prime\prime}}_{\infty}\leq{C}\norm{ G}_{\infty}
\end{align*}
thanks to $\norm{\calA^2G}_{\infty}\leq \norm{\calA G}_{\infty}\leq\norm{G}_{\infty}$ and
$\norm{\calA \tilde{S}}_{\infty}\leq\norm{\tilde{S}}_{\infty}$.
\begin{mhchange}
Since $\tilde{S}$ does not belong to $\fspace{X}$, we
now define
\begin{align}
\label{Lem:LinSolOp.Eqn3}
S\at{x}:=\tilde{S}\at{x}-\tilde{S}\at{0}-f_1\frac{ \sqrt{2\pi}}{m^\prime\at{k_c}}\bat{\cos\at{k_cx}-1}-
f_2\frac{ \sqrt{2\pi}}{m^\prime\at{k_c}}\sin\at{k_cx}\,
\end{align}
as well as
\begin{align}
\label{Lem:LinSolOp.Eqn4}
\eta:=\at{1-c^2}\at{f_1\frac{ \sqrt{2\pi}}{m^\prime\at{k_c}}-\tilde{S}\at{0}}\,,
\end{align}
and observe that $S\in\space{X}$ and
\eqref{Lem:LinSolOp.Eqn0} hold by construction.
Moreover, $S$ and $\eta$ depend linearly on $G$ and
the above estimates for $f_1$, $f_1$ and $\tilde{S}$ provide the desired estimates for both $S$ and $\eta$.
Finally, the uniqueness of $\pair{S}{\eta}$ is a direct consequence of ${S}\in\fspace{X}$ and
Lemma \ref{Rem:MProps}.
\end{mhchange}
\end{proof}
\begin{mhchange}
Notice that the solution $\pair{S}{\eta}$ to \eqref{FPP:CorrEqn}$_1$ is unique only in the space $\fspace{X}\times\Rset$
and that further solution branches exists due to the nontrivial kernel functions of $\calM$. For instance, replacing \eqref{Lem:LinSolOp.Eqn3} and \eqref{Lem:LinSolOp.Eqn4} by
\begin{align*}
S\at{x}:=\tilde{S}-\tilde{S}\at{0}\,,\qquad
\eta:=-\at{1-c^2}\tilde{S}\at{0}
\end{align*}
we can define an operator
\begin{align}
\label{Eqn:AlternativeL}
\bar{\calL}\colon G\in\fspace{Y}\to \pair{\bar{S}}{\bar\eta}\in\bar{\fspace{X}}\times\Rset\,,\qquad
\bar{\fspace{X}}:=\big\{\bar{S}\in\fspaceW^{2,\infty}\at\Rset\;:\; S\at{0}=0\big\}\,,
\end{align}
which provides another solution to the affine problem \eqref{FPP:CorrEqn}$_1$. The corresponding corrector $S$,
however, does in general not belong to $\fspace{X}$ as it is oscillatory for both $x\to-\infty$ and $x\to+\infty$.
\par
We emphasise that the three-parameter family of travelling waves
$R=R_0+S$, which we construct below by fixed points arguments involving $\calL$,
is -- at least for sufficiently small $\delta$ -- independent of the
details in the definition of $\calL$. The reason is, roughly speaking, that changing $\calL$ is equivalent to changing
$R_0$, see the discussion at
the end of Section \ref{sect:uniqueness}. However, choosing $\fspace{X}\times\Rset$ as image space for $\calL$
provides more information on the resulting family of travelling waves:
The existence of $\lim_{x\to+\infty}S\at{x}$ reveals that
for each $c$ there exists one wave $R=R_0+S$ that complies with the causality principle as
it is non-oscillatory for $x\to+\infty$.
\end{mhchange}
%
\subsection{Properties of the nonlinear operator \texorpdfstring{$\calG$}{}}
%
In order to investigate the properties of the nonlinear superposition operator $\calG$, we introduce
a class of admissible perturbations $S$. More precisely, we
say that $S\in\fspace{X}$ is  \emph{$\delta$-admissible}
if there exist two number $x_-<0<x_+$, which both depend on $S$ and $\delta$, such that
\begin{enumerate}
\item
$R_0\at{x_\pm}+S\at{x_\pm}=\pm\delta$\,,
\item
$R_0\at{x}+S\at{x}<-\delta \quad \text{for}\quad x<x_-$\,,
\item
$R_0\at{x}+S\at{x}>+\delta \quad \text{for}\quad x>x_+$\,,
\item
$\tfrac12R_0^\prime\at{0}<R_0^\prime\at{x}+S^\prime\at{x}<2R_0^\prime\at{0} \quad \text{for}\quad x_-<x<x_+$\,,
\end{enumerate}
where $R_0$ is the chosen wave for $\delta=0$. Below we show that each sufficiently small ball in $\fspace{X}$ consists entirely of
$\delta$-admissible functions, and this enables us to find travelling waves by
the contraction mapping principle.
\bigpar
We are now able to derive the second key argument for our fixed-point argument.
\begin{lemma}
\label{Lem.G.Props.A}
Let $S\in\fspace{X}$ be $\delta$-admissible and $G=\calG\at{S}$ as in \eqref{Eqn:Def.G}. Then we have
\begin{align*}
\norm{G}_\infty\leq {C}\,,\qquad \supp G \subseteq \ccinterval{-C\delta}{C\delta}\,,\qquad
\int_\Rset G\at{x}\dint{x}\leq {C}\at{1+\norm{S^{\prime\prime}}_{\infty}}\delta^2
\end{align*}
for some constant $C$ independent of $S$ and $\delta$.
\end{lemma}
\begin{proof}
The first assertion is a consequence of $\norm{G}_\infty\leq 1+ C_\Psi$. Since $S$ is $\delta$-admissible,
we have
\begin{align*}
\supp G = \ccinterval{x_-}{x_+}\,,\qquad
\pm\delta=\pm\int_0^{x_\pm}\bat{R_0^\prime\at{x}+S^\prime\at{x}}\dint{x}
\end{align*}
with $x_\pm$ as above, and this implies
\begin{align}
\label{Lem.G.Props.A.Eqn0}
\frac{1}{2R_0^\prime\at{0}}\delta\leq\abs{x_\pm}\leq\frac{2}{R_0^\prime\at{0}}\delta\,,\qquad
\supp  G\subseteq \frac{2}{R_0^\prime\at{0}}\ccinterval{-\delta}{\delta}\,.
\end{align}
Using the Taylor estimate
\begin{align}
\label{Lem.G.Props.A.Eqn1}
\babs{R_0^\prime\at{x}+S^\prime\at{x}
-R_0^\prime\at{0}-S^\prime\at{0}}
\leq\at{\norm{R_0^{\prime\prime}}_{\infty}+\norm{S^{\prime\prime}}_{\infty}}\abs{x}\,,
\end{align}
we also verify that
\begin{align}
\label{Lem.G.Props.A.Eqn2}
\abs{x_\pm\mp\frac{\delta}{R_0^\prime\at{0}+S^\prime\at{0}}}
\leq
\frac{\abs{x_\pm}^2}{2}\frac{
\norm{R_0^{\prime\prime}}_\infty+\norm{S^{\prime\prime}}_\infty}{R_0^\prime\at{0}+S^\prime\at{0}}
\leq {4\delta^2}\frac{
\norm{R_0^{\prime\prime}}_\infty+\norm{S^{\prime\prime}}_\infty}{R_0^\prime\at{0}^3}\,.
\end{align}
A direct computation now yields
\begin{align}
\label{Lem.G.Props.A.Eqn3}
\int_\Rset G\at{x}\dint{x}&=\int_{x_-}^{x_+}
\Psi_\delta^\prime\bat{R_0\at{x}+S\at{x}}\dint{x}-\int_{x_-}^{x_+}\sgn\bat{R_0\at{x}}\dint{x}
=\int_{-\delta}^\delta \Psi_\delta^\prime\at{r}\frac{\dint{r}}{z\at{r}}-\abs{x_++x_-}\,,
\end{align}
due to $\sgn\bat{R_0\at{x}}=\sgn\at{x}$. Here,
the function $z$ with
$z\bat{R_0\at{x}+S\at{x}}=R_0^\prime\at{x}+S^\prime\at{x}$ for all $x\in\ccinterval{x_-}{x_+}$
is well-defined since $R+S_0$ is strictly increasing on $\ccinterval{x_-}{x_+}$. 
Thanks to \eqref{Lem.G.Props.A.Eqn1},
our assumption $I_\delta=\int_{-\delta}^{+\delta}\Psi_\delta^\prime\at{r}\dint{r}=0$,
and the estimate $z\at{r},z\at{0}\geq \tfrac12 R_0^\prime\at{0}$
we get
\begin{align*}
\abs{\int_{-\delta}^\delta \Psi_\delta^\prime\at{r}\frac{\dint{r}}{z\at{r}}} &=
\abs{\int_{-\delta}^\delta \Psi_\delta^\prime\at{r}\at{\frac{1}{z\at{r}}-\frac{1}{z\at{0}}}\dint{r}}
\leq
\int_{-\delta}^\delta \abs{\Psi_\delta^\prime\at{r}}\frac{\abs{z\at{r}-z\at{0}}}{z\at{r}z\at{0}}\dint{r}
\\&\leq
C\delta\frac{\bat{\abs{x_+}+\abs{x_-}}\bat{
\norm{R_0^{\prime\prime}}_\infty+\norm{S^{\prime\prime}}_\infty}}{R_0^\prime\at{0}^2}\,.
\end{align*}
and combining this with \eqref{Lem.G.Props.A.Eqn0}, \eqref{Lem.G.Props.A.Eqn2} and \eqref{Lem.G.Props.A.Eqn3} gives
\begin{align}
\label{Lem.G.Props.A.Eqn4}
\abs{\int_\Rset G\at{x}\dint{x}}&
\leq
\abs{x_-+x_+}+C\delta\bat{\abs{x_-}+\abs{x_+}}\frac{
\norm{R_0^{\prime\prime}}_\infty+\norm{S^{\prime\prime}}_\infty}{R_0^\prime\at{0}^2}
\leq
C\delta^2\frac{
\norm{R_0^{\prime\prime}}_\infty+\norm{S^{\prime\prime}}_\infty}{R_0^\prime\at{0}^3}\,.
\end{align}
By Proposition~\ref{Prop:ExistenceOfAnchor}~\ref{it:anch-two}, $R_0^\prime\at{0}$ is bounded from below. Moreover, combining Proposition~\ref{Prop:ExistenceOfAnchor}~\ref{it:anch-one} with the equation for $R_0^{\prime\prime}$, that is
\begin{equation*}
c^2 R_0^{\prime\prime} = \Delta_1 R_0 - \Delta_1 \sgn{}\,,
\end{equation*}
we find a constant $C$, which depends only on $c_0$ and $c_1$, such that
$\norm{R_0^{\prime\prime}}_\infty \leq C$.
The second and third assertion are now direct consequences of these observations and the estimates \eqref{Lem.G.Props.A.Eqn2} and \eqref{Lem.G.Props.A.Eqn4}.
\end{proof}
\begin{corollary}
\label{Cor:G.Props.1}
There exists a constant $C_\calG$, which is independent of $\delta$, such that
\begin{align}
\label{Cor:G.Props.1.Eqn1}
\norm{\calA{G}}_\infty\leq
C_\calG\delta\,,
\qquad
\norm{\calA^2{G}}_\infty\leq
C_\calG\at{1+\norm{S^{\prime\prime}}_\infty}\delta^2\,,
\end{align}
holds with $G=\calG\at{S}$ for all $\delta$-admissible $S$.
\end{corollary}
\begin{proof}
Thanks to Lemma \ref{Lem.G.Props.A} and since $\calA$ is the convolution with the characteristic function of the interval $\ccinterval{-\tfrac12}{+\tfrac12}$, there exists a constant $C$ such that
\begin{align*}
\begin{array}{lclcl}
\displaystyle\abs{\calA G\at{x}}&\leq& C\delta &\;\text{  for  }\;& \nabs{x\pm\tfrac12}\leq C\delta\,,\\
\displaystyle\calA G\at{x}&=&\int_\Rset  G\at{x}\dint{x} &\;\text{  for  }\;& \abs{x}\leq \tfrac12-C\delta\,,\\
\displaystyle\calA G\at{x}&=&0 &\;\text{  for  }\;& \abs{x}\geq \tfrac{1}{2}+C\delta\,,
\end{array}
\end{align*}
see Figure \ref{Fig:g_fct} for an illustration.
The first bound in \eqref{Cor:G.Props.1.Eqn1} is now a consequence of  the trivial estimate  $\abs{\int_\Rset G\at{x}\dint{x}}\leq  \abs{\supp G}\norm{G}_\infty\leq {C}\delta$,  whereas the second one follows from
\begin{align*}
\abs{\at{\calA^2G}\at{x}}\leq C\delta^2 + \abs{\int_\Rset G\at{x}\dint{x}}\quad\text{for all}\quad x\in\Rset
\end{align*}
and the refined estimate $\abs{\int_\Rset G\at{x}\dint{x}}\leq {C}\at{1+\norm{S}^{\prime\prime}_\infty}\delta^2 $.
\end{proof}
\begin{figure}[ht!]%
\centering{%
\includegraphics[width=0.9\textwidth, draft=\figdraft]%
{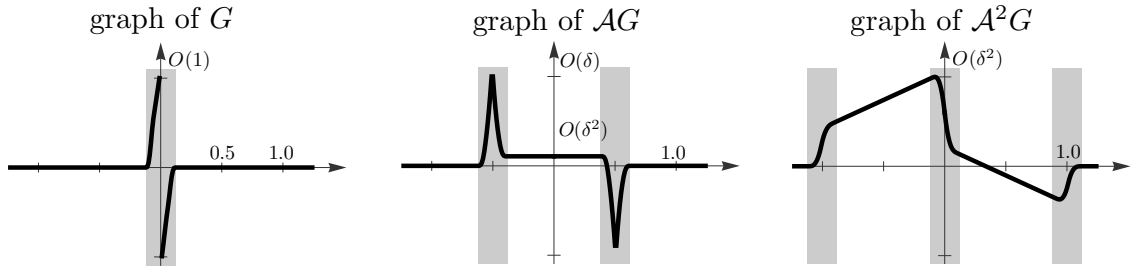}%
}%
\caption{%
Properties of $G=\calG\at{S}$ for $\delta$--admissible $S$. The shaded regions indicate intervals with length of order $\DO{\delta}$.
}%
\label{Fig:g_fct}%
\end{figure}
In the general case $I_\delta\neq0$, one finds -- due to
$\int_\Rset G\at{x}\dint{x}= 2I_\delta+\nDO{\delta^2}$ -- the weaker estimate
$\norm{\calA^2 G}_\infty\leq{C}\at{1+\norm{S^{\prime\prime}_\infty}}\delta$.
This bound is still sufficient to establish a fixed point argument, but provides only
a corrector $S$ of order $\DO{\delta}$. Recall, however, that Lemma \ref{Lem:Trafo} shows  that
shifting $\Psi_\delta$ and changing $R_0$ allows us to find correctors of order $\nDO{\delta^2}$ even in the case
$I_\delta\neq0$.
\bigpar
We finally derive continuity estimates for $\calG$.
\begin{lemma}
\label{Lem:Lip}
There exists a constant $C_L$ independent of $\delta$ such that
\begin{align*}
\norm{\calA^2 G_2- \calA^2 G_1}_\infty+
\norm{\calA G_2- \calA G_1}_\infty+\delta \norm{G_2- G_1}_\infty \leq
 C_L \delta \norm{S_2^\prime-S_1^\prime}_\infty
\end{align*}
holds for all $\delta$-admissible  correctors $S_1$  and $S_2$ with $G_\ell=\calG\at{S_\ell}$.
\end{lemma}
\begin{proof} According to Lemma \ref{Lem.G.Props.A}, there exists a constant $C$, such that
$G_\ell\at{x}=0$ for all $x$ with $\abs{x}\geq{C}\delta$. For $\abs{x}\leq {C}\delta$, we
use Taylor expansions for $S_1-S_2$ at $x=0$ to find
\begin{align*}
\abs{S_2\at{x}-S_1\at{x}}\leq \norm{S_2^\prime-S_1^\prime}_\infty\abs{x}\,,
\end{align*}
where we used that $S_2\at{0}-S_1\at{0}=0$. Combining this estimate with the upper bounds for
$\Psi_\delta^{\prime\prime}$  gives
\begin{align*}
\babs{G_2\at{x}-G_1\at{x}}\leq{ \frac{C}{\delta}\babs{S_2\at{x}-S_1\at{x}}}
\leq C \norm{S_2^\prime-S_1^\prime}_\infty\,
\end{align*}
for all $\abs{x}\leq{C}\delta$, and this implies the desired estimate for $\norm{G_2-G_1}_\infty$. We also have
\begin{align*}
\norm{\calA^2 G_2- \calA^2G_1}_\infty\leq \norm{\calA G_2- \calA G_1}_\infty\leq\abs{\supp \at{G_2-G_1}}\norm{G_2-G_1}_\infty\leq{C}\delta\norm{G_2-G_1}_\infty\,,
\end{align*}
and this completes the proof.
\end{proof}
%
%
%
\subsection{Fixed point argument}
%
\begin{mhchange}
Now we have prepared all ingredients to prove that the operator
\begin{align*}
\calT:=\calP_{S}\circ\calL\circ\calG
\end{align*}
admits a unique fixed point in the space
\end{mhchange}
\begin{align*}
\fspace{X}_\delta:=\Big\{ S\in\fspace{X}\;:\; \norm{S}_\infty \leq{C_0}\delta^2\,,\quad
\norm{S^\prime}_\infty \leq{C_1}\delta\,,\quad
\norm{S^{\prime\prime}}_\infty \leq{C_2}
\Big\}\,.
\end{align*}
\begin{mhchange} Here, $\calP_{S}$ denotes the projector on the first component, that means $\calP_S\pair{S}{\eta}={S}$, \end{mhchange} and
the constants $C_i$ are defined by
\begin{align*}
C_2 := C_{\calM}\at{1+C_\Psi},\,\qquad
C_1:=C_{\calM}C_{\calG}\,,\qquad
C_0:=C_{\calM}C_{\calG}(1+C_2)\,.
\end{align*}
Notice that any fixed point of $\calT$ provides a solution to \eqref{FPP:CorrEqn} and vice versa.
\begin{lemma}
\label{Lem:SelfMapping}
\begin{mhchange}
For all sufficiently small $\delta$,
the operator $\calT$ has a unique fixed point in $X_\delta$.
\end{mhchange}
\end{lemma}
\begin{proof}
\emph{\underline{Step 1:}} We first show that each $S\in\fspace{X}_\delta$ is $\delta$-admissible provided that $\delta$ is sufficiently small. According to Proposition \ref{Prop:ExistenceOfAnchor}, there exist positive constants
$r_0$, $x_0$, and $d_0$ such that
\begin{align*}
\babs{R_0\at{x}}\ge r_0 \quad \text{for}\quad \abs x>x_0\,,
\qquad
d_0<R_0^\prime\at{x} \quad \text{for}\quad \abs x<x_0\,,
\end{align*}
and combining the upper estimate for $\norm{R_0}_\infty$ with the equation for $R_0$ we find
$\norm{R_0^{\prime\prime}}_\infty\leq D_2$ for some constant $D_2$.
We now set
\begin{align*}
\delta_0:=\frac12\min\left\{\frac{d_0}{2\tfrac{D_2}{d_0}+C_1},\,
x_0d_0,\,
\sqrt{\frac{r_0}{C_0}},\,r_0
\right\},\qquad
x_\delta:=\frac2{d_0}\delta
\end{align*}
and assume that $\delta<\delta_0$.
For any $x$ with $\abs x\le x_\delta\leq x_0$, we then estimate
\begin{align*}
\abs{R_0^\prime\at{x}+S^\prime\at{x}-R_0^\prime\at{0}}
\le D_2x_\delta+C_1\delta
\le \at{2\tfrac{D_2}{d_0}+C_1}\delta
\le \tfrac12d_0<\tfrac12R_0^\prime\at{0}\,,
\end{align*}
and this gives $\tfrac{1}{2}R_0^\prime\at{0}\leq R_0^\prime\at{x}+S^\prime\at{x}\leq
\tfrac{3}{2}R_0^\prime\at{0}$. Moreover, $x_\delta\leq \abs{x}\leq x_0$
implies
\begin{align*}
\babs{R_0\at{x}+S\at{x}}\geq
\abs{\int_0^{x}R_0^\prime\at{s}\dint{s}}-\norm{S^\prime}_\infty\abs{x}
\ge(d_0-C_1\delta)\abs{x}>
\tfrac12d_0\cdot\frac2{d_0}\delta=\delta\,,
\end{align*}
whereas for $\abs{x}>x_0$ we find
\begin{align*}
\babs{R_0\at{x}+S\at{x}}
\geq r_0-C_1\delta^2\geq\tfrac12 r_0\geq\delta\,.
\end{align*}
Using
\begin{align*}
x_-:=\max\{x\;:\;R_0\at{x}+S\at{x}\leq-\delta\}\,,\qquad
x_+:=\min\{x\;:\;R_0\at{x}+S\at{x}\geq+\delta\}\,,
\end{align*}
we now easily verify that $S$ is $\delta$-admissible.
\par
\emph{\underline{Step 2:}}
We next show that $\calT\at{\fspace{X}_\delta}\subset\fspace{X}_\delta$ for all $\delta<\delta_0$. Since each
$S\in\fspace{X}_\delta$ is $\delta$-admissible, Corollary \ref{Cor:G.Props.1} yields
\begin{align*}
\norm{\calA\calG\at{S}}_\infty\leq C_\calG\delta\,,\qquad
\norm{\calA^2\calG\at{S}}_\infty\leq C_\calG\at{1+ C_2}\delta^2
\end{align*}
and $\norm{\calG\at{S}}_\infty\leq 1+ C_\Psi$ holds by definition of $\calG$ and Assumption \ref{Ass:DPsi}.
Lemma \ref{Lem:LinSolOp} now provides
\begin{align*}
\begin{array}{lclcl}
\norm{\calT\at{S}}_\infty&\leq& C_\calM C_\calG\at{1+C_2}\delta^2&=&C_0\delta^2\,,\\
 \norm{\calT\at{S}^\prime}_\infty&\leq& C_\calM C_\calG\delta &=&C_1\delta\,,\\
\norm{\calT\at{S}^{\prime\prime}}_\infty&\leq& C_\calM \at{1+C_\Psi}&=&C_2\,,
\end{array}
\end{align*}
and hence $\calT\at{S}\in\fspace{X}_\delta$.
\par
\emph{\underline{Step 3:}}
We equip $\fspace{X}_\delta$ with the norm
$\norm{S}_{\#}=\norm{S}_\infty+\norm{S^\prime}_\infty + \delta\norm{S^{\prime\prime}}_\infty$,
which is, for fixed $\delta$, equivalent to the standard norm. For given $S_1, S_2\in\fspace{X}_\delta$,
we now employ the estimates from Lemma \ref{Lem:LinSolOp} and Lemma \ref{Lem:Lip} for $S=S_2-S_1$
and $G=\calG\at{S_2}-\calG\at{S_1}$. This gives
\begin{align*}
\norm{\calT\at{S_2}-\calT\at{S_1}}_{\#} &\leq C_\calM\Bat{
\norm{\calA^2 \calG\at{S_2}-\calA^2\calG\at{S_1}}_\infty +
\norm{\calA \calG\at{S_2}-\calA\calG\at{S_1}}_\infty+
\delta \norm{\calG\at{S_2}-\calG\at{S_1}}_\infty}\\
&
\leq C_\calM C_L \delta \norm{S_2^\prime-S_1^\prime}_{\infty}
\leq C_\calM C_L \delta \norm{S_2-S_1}_{\#}
\,,
\end{align*}
and we conclude that $\calT$ is contractive with respect to $\norm{\cdot}_{\#}$ provided that
$\delta<1/\at{C_\calM C_L}$. The claim is now a direct consequence of the Banach Fixed Point Theorem.
\end{proof}

\begin{mhchange}
The previous result implies the existence of a three-parameter family of waves that is parametrised by the speed $c\in\ccinterval{c_0}{c_1}$ and by $R_0$, where $R_0$ can be regarded as parameter in the two-dimensional   $\fspaceL^\infty$-kernel of $\calM$. 
\begin{proposition}
\label{Thm:Existence}
Suppose that $I_\delta=0$ for all $\delta$. Then there exists $\delta_0>0$ with the following property: For any $\delta<\delta_0$, each $c\in\ccinterval{c_0}{c_1}$, and any $R_0$ as in Proposition \ref{Prop:ExistenceOfAnchor}
there exists an $\delta$-admissible
\begin{align*}
S\in\fspace{X}_\delta\,\cap\,\at{\fspaceL^2\at\Rset\oplus\xspan{\Big\{1,\,Y_1 - \frac{ \sqrt{2\pi}}{m^\prime\at{k_c}} \cos\at{k_c\cdot}\,,\,Y_2- \frac{ \sqrt{2\pi}}{m^\prime\at{k_c}} \sin\at{k_c\cdot} \Big\}}}\,. 
\end{align*}
such that $R=R_0+S$ and solves the travelling wave equation \eqref{Eqn:TW2} for some $\mu$.
In particular, we have $R\at{0}=0$, the limits
\begin{align*}
\lim_{x\to-\infty}\bat{R\at{x}-\al_-\cos\at{k_cx}-\beta_-\sin\at{k_c{x}}}\qquad\text{and}\qquad
\lim_{x\to+\infty}\bat{R\at{x}-R_0\at{x}}
\end{align*}
are well-defined for some constants $\al_-$, $\be_-$ depending on $c$ and $R_0$,
and the estimates
\begin{align}
\label{Thm:FixedPoints.Eqn2}
R\at{x}\leq -\delta\quad{\text{for}}\quad x\leq-{C}\delta\,,\qquad
R\at{x}\geq +\delta\quad{\text{for}}\quad x\geq+{C}\delta
\end{align}
hold for some constant $C>0$ independent of $c$ and $R_0$.
\end{proposition}
\begin{proof}
For given $c$ and $R_0$, Lemma \ref{Lem:SelfMapping} provides a unique fixed point $S\in{X}_\delta$
of $\calT$, which solves
\begin{align*}
\calM S = \calA^2\calG\at{S}+\eta
\end{align*}
for some $\eta\in\Rset$, and this implies that $R=R_0+S$ is in fact a travelling wave. Moreover, by construction --
see the proof of Lemma \ref{Lem:LinSolOp} -- we also have
\begin{align*}
S=Z+\la +f_1 \bat{ Y_1 - \frac{\sqrt{2\pi}}{m^\prime\at{k_c}} \cos\at{k_c\cdot}}+ f_2 \bat{Y_2 -  \frac{\sqrt{2\pi}}{m^\prime\at{k_c}} \sin\at{k_c\cdot}}
\end{align*}
for some constants $f_1$, $f_2$ and $\la$ and a function
$Z\in\fspaceL^2\at{\Rset}$ with $Z\at{x}\to0$ as $x\to\pm\infty$.
The claims on the asymptotic behaviour as $x\to\pm\infty$ now follow immediately since
$R_0$ has harmonic tails oscillations with wave number $k_c$. Finally,
the fixed point $S$ is $\delta$-admissible -- see the proof of Lemma \ref{Lem:SelfMapping} --
and this implies the validity of \eqref{Thm:FixedPoints.Eqn2} due to $0\leq x_+, -x_-\leq C\delta$.
\end{proof}
Notice that Proposition~\ref{Thm:Existence} yields a genuine
three-parameter family in the sense that
different choices of the parameters $c$ and $R_0$ correspond to
different tail oscillations for $x\to+\infty$ and hence to different waves $R=R_0+S$. 
This finishes the existence proof of Theorem~\ref{Thm:Main}.
\end{mhchange}
%

%
%
\section{Uniqueness of phase transition waves}\label{sect:uniqueness}
%
%
%
\begin{mhchange}
In this section we establish the uniqueness result of the Theorem~\ref{Thm:Main}
by showing
that the family provided by Proposition \ref{Thm:Uniqueness} contains all
phase transition waves that have harmonic tails oscillations for $x\to+\infty$ and
penetrate the spinodal region in a small interval only.
\begin{lemma}
\label{Lem:Uniqueness}
Let $\kappa>\tfrac12$ be given and suppose that $I_\delta=0$ for all $\delta$. Then
there exists $\delta_\kappa>0$ such that
the following statement holds for all $0<\delta<\delta_\kappa$: Let
$\pair{R_1}{\mu_1}$ and $\pair{R_2}{\mu_2}$ be two solutions to the travelling wave equation \eqref{Eqn:TW2} with speed $c\in\ccinterval{c_0}{c_1}$ such that
\begin{align*}
R_i\in\fspaceW^{2,\infty}\at{\Rset}\,,\qquad\mu_i\in\Rset\,,\qquad
R_i\at{0}=0
\end{align*}
and
\begin{align*}
R_i\at{x}\leq -\delta\quad{\text{for}}\quad x\leq-\delta^\kappa\,,\qquad
R_i\at{x}\geq +\delta\quad{\text{for}}\quad x\geq+\delta^\kappa
\end{align*}
for both $i=1$ and $i=2$. Then, $R_1$ and $R_2$ are either identical or satisfy
\begin{align*}
 R_1\at{x}-R_2\at{x}-\alpha_+\bat{\cos\at{k_cx}-1}-\beta_+\sin\at{k_cx}-\ga_+\quad\xrightarrow{\;x\to+\infty\;}\quad0
\end{align*}
for some constants $\ga_+$ and $\pair{\alpha_+}{\beta_+}\neq\pair{0}{0}$.
\end{lemma}
\begin{proof}
For given $R_1$, $R_2$, there exist constant $\mu_1,\mu_2\in\Rset$ such that
\begin{align*}
\calM\at{R_2-R_1}=\calA^2 G+\mu_2-\mu_1\,,\qquad
G:=\Psi_\delta^\prime\at{R_2}-\Psi_\delta^\prime\at{R_1}\,.
\end{align*}
By assumption and due to the bounds of $\Psi^{\prime\prime}_\delta$ we also
find $G\at{x}=0$ for $\abs{x}\geq \delta^\kappa$ as well as
\begin{align*}
\babs{G\at{x}}\leq\frac{C}{\delta}\abs{R_2\at{x}-R_1\at{x}}\leq C\delta^{\kappa-1}
\norm{R_2^\prime-R_1^\prime}_{\infty}\quad\text{for}\quad \abs{x}\leq \delta^\kappa\,,
\end{align*}
and this implies
\begin{align*}
\norm{\calA G}_\infty\leq \abs{\supp G}\norm{G}_\infty\leq C\delta^{2\kappa-1}\norm{R_2^\prime-R_1^\prime}_{\infty}\,.
\end{align*}
Moreover, Lemma \ref{Lem:LinSolOp} provides $S\in\fspace{X}$ as well as $\eta\in\Rset$ such that
\begin{align*}
\calM{S}=\calA^2G+\eta\,,\qquad
\norm{S^\prime}_\infty\leq {C}\delta^{2\kappa-1}
\norm{R_2^\prime-R_1^\prime}_{\infty}.
\end{align*}
In particular, we have
\begin{align*}
\calA^2G=\calM\bat{R_2-R_1-\at{1-c^2}^{-1}\at{\mu_2-\mu_1}}=
\calM\at{S-\at{1-c^2}^{-1}\eta}.
\end{align*}
Since the space of bounded kernel functions for $\calM$ is spanned by $\sin\at{k_c\cdot}$ and $\cos\at{k_c\cdot}$, we  conclude that there exist constants $\alpha_+$ and $\beta_+$ such that
\begin{align*}
R_2\at{x}-R_1\at{x}=S\at{x}-\sigma+\alpha_+ \at{1-\cos\at{k_c{x}}}+\beta_+ \sin\at{k_c{x}}+\gamma_+\,,
\end{align*}
where $\sigma:=\lim_{x\to+\infty}S\at{x}$ and $\gamma_+:=\at{1-c^2}^{-1}\at{\mu_2-\mu_1-\eta}+\sigma-\alpha_+$.
In case of $\alpha_+=\beta_+=0$ we therefore find
\begin{align*}
\norm{R_2^\prime-R_1^\prime}_\infty=\norm{S^\prime}_\infty\leq
C\delta^{2\kappa-1}
\norm{R_2^\prime-R_1^\prime}_{\infty}\,,
\end{align*}
and combining this with $R_1\at{0}=R_2\at{0}$ we get $R_2=R_1$ for all  sufficiently small $\delta$.
\end{proof}
\begin{proposition}
\label{Thm:Uniqueness}
Suppose that $I_\delta=0$ for all $\delta$ and that $\kappa$ with $\tfrac{1}{2}<\kappa<1$ is fixed.
Then there exists $\delta_\kappa$ with $0<\delta_\kappa\leq\delta_0$
such that the following statement holds for all $0<\delta<\delta_\kappa$: Let $R$ be a travelling waves with speed $c\in\ccinterval{c_0}{c_1}$ such that
the limit
\begin{align*}
\lim\limits_{x\to+\infty}\bat{R\at{x}-R_0\at{x}}
\end{align*}
is well-defined for some $R_0$ from Proposition \ref{Prop:ExistenceOfAnchor} and
such that
\begin{align*}
R\at{x}\leq -\delta\quad{\text{for}}\quad x\leq-\delta^\kappa\,,\qquad
R\at{x}\geq +\delta\quad{\text{for}}\quad x\geq+\delta^\kappa.
\end{align*}
Then $R$ belongs to the family of waves provided by Proposition \ref{Thm:Existence}.
\end{proposition}
\begin{proof}
Let $R_0+S$ be the travelling wave from Proposition \ref{Thm:Existence}. By construction, $R-R_0-S$
converges as $x\to+\infty$ and
for all sufficiently small $\delta$ we also
have $C\delta\leq\delta^\kappa$. Lemma~\ref{Lem:Uniqueness} applied with
$R_1=R$ and $R_2=R_0+S$ therefore implies $R=R_0+S$.
\end{proof}
With Proposition~\ref{Thm:Existence} and Proposition~\ref{Thm:Uniqueness} we have established our existence and uniqueness result in the special case that $I_\delta=0$ holds for for all $\delta$. The corresponding result for
the general case is then provided by Lemma \ref{Lem:Trafo}.
\par
We finally mention a particular consequence of our uniqueness result, namely that the
family from Proposition \ref{Thm:Existence} does not depend on the particular choice of the solution
operator $\calL$ to the affine problem \eqref{FPP:CorrEqn}$_1$. At a first glance, this might
be surprising since the operator $\calT$ and hence each fixed point surely depend on $\calL$.
We can, however, argue as follows. Suppose we would choose in the proof of Lemma \ref{Lem:LinSolOp} another
reasonable solution operator $\bar{\calL}$ (for instance, the operator
from \eqref{Eqn:AlternativeL} that does not involve any kernel function of $\calM$).
Repeating all arguments from Section~\ref{sect:proof} we then find -- for any given $\delta$, $c$, and $R_0$ -- a different corrector
$\bar{S}\in\fspaceW^{2,\infty}\at\Rset$. In general, this corrector $\bar{S}$ does not
converge as $x\to+\infty$ but satisfies
\begin{align*}
\bar{S}\at{0}=0\,, \qquad
\norm{\bar{S}}_\infty\leq\bar{C}\delta^2\,,\qquad
\norm{\bar{S}^\prime}_\infty\leq\bar{C}\delta\,,\qquad
\norm{\bar{S}^{\prime\prime}}_\infty\leq\bar{C}
\end{align*}
for some constant $\bar{C}$ that is independent of $c$, $R_0$, and $\delta$.
Moreover, we also have
\begin{align*}
\bar{S}\in\fspaceL^2\at\Rset\oplus \xspan\big\{1,\,Y_1,\,Y_2,\,\cos\at{k_c\cdot},\,\sin\at{k_c\cdot}\big\}
\end{align*}
that means the tail oscillations of $\bar{S}$ for both $x\to-\infty$ and $x\to+\infty$
are again harmonic waves with wave number $k_c$. Adding
a suitable linear combination of $1-\cos\at{k_c\cdot}$ and $\sin\at{k_c\cdot}$ to $R_0$
we can construct another wave $\bar{R}_0$ such that
$\bar{R}_0$ and $R_0+\bar{S}$ have the same tails oscillations as $x\to+\infty$.
This function $\bar{R}_0$ is, at least for small $\delta$, also a travelling wave for
the unperturbed problem and hence among the family of waves provided by
Proposition \ref{Prop:ExistenceOfAnchor}. We can therefore use $\bar{R}_0$ instead of
$R_0$ in order to define the operator $\calG$.
Theorem \ref{Lem:LinSolOp}, which relies on the oscillation-preserving operator $\calL$, then
provides a corrector $S$ that converges as $x\to+\infty$, and
from Lemma \ref{Lem:Uniqueness} we finally infer that $\bar{R}_0+S=R_0+\bar{S}$
because both waves have, by construction, the same tail oscillations for $x\to+\infty$.
We therefore conclude, at least for small $\delta$,
that changing $\calL$ does not alter
the family of travelling waves but only its
parametrisation by $R_0$.
\end{mhchange}
%
%
\section{Kinetic relations}\label{sect:kinrel}
%
We finally show that the kinetic relation does not change to order $\nDO{\delta}$. To this end we denote
by $R_\delta$ a travelling wave solution to \eqref{eq:strain} as provided by Theorem \ref{Thm:Main}. The corresponding \emph{configurational force}, cf. \cite{Herrmann:10b},
is then defined by $\Upsilon_\delta := \Upsilon_{\mathrm{e},\delta}-\Upsilon_{\mathrm{f},\delta}$ with
\begin{align*}
\Upsilon_{\mathrm{e},\delta}:=\Phi_\delta\at{\bar{r}_{\delta,+}}-\Phi_\delta\at{\bar{r}_{\delta,-}}\,,\qquad
\Upsilon_{\mathrm{f},\delta}:=
\frac{\Phi^\prime_\delta\at{\bar{r}_{\delta,+}}+
\Phi^\prime_\delta\at{\bar{r}_{\delta,-}}}{2}\Bat{\bar{r}_{\delta,+}-\bar{r}_{\delta,-}}\,,
\end{align*}
where the \emph{macroscopic} strains $\bar{r}_{\delta,\pm}$ on both sides of the interface can
be computed from $R_\delta$ via
\begin{align*}
\bar{r}_{\delta,\pm}=\lim_{L\to\infty}\frac{1}{L}\int_0^{+L}{R_\delta\at{\pm x}}\dint{x}\,.
\end{align*}
\begin{lemma}
Let $R_\delta$ be a travelling wave from Theorem \ref{Thm:Main}, and $R_0$ the corresponding wave for $\delta=0$. Then we have
$\Upsilon_{\delta}=\Upsilon_0+\nDO{\delta^2}$.
\end{lemma}
\begin{proof}
By construction, we know that the only asymptotic contributions to the profile $R_\delta$
are due to $R_0-I_\delta$ plus a small asymptotic corrector of order $\nDO{\delta^2}$
from $\xspan\big\{1,\,Y_1,\,Y_2\big\}$. This implies
\begin{align*}
\bar{r}_{\delta,\pm}=\bar{r}_{0,\pm}-I_\delta + \nDO{\delta^2}\,.
\end{align*}
As $\bar{r}_{0,\pm}$ and $\bar{r}_{\delta,\pm}$ are both larger than $\delta$ we know that
\begin{align*}
\Psi^\prime_\delta\at{\bar{r}_{\delta,\pm}}=\mp1=\Psi^\prime_0\at{\bar{r}_{\delta,\pm}}\,.
\end{align*}
Thus, we conclude
\begin{align*}
\qquad \Phi^\prime_\delta\at{\bar{r}_{\delta,\pm}}=\bar{r}_{\delta,\pm}\mp1=\Phi^\prime_0\at{\bar{r}_{0,\pm}}-I_\delta+\nDO{\delta^2}\,,
\end{align*}
and hence
\begin{align*}
\Upsilon_{\mathrm{f},\delta}=\Upsilon_{\mathrm{f},0}-I_\delta\at{\bar{r}_{0,+}-\bar{r}_{0,-}}+\nDO{\delta^2}\,.
\end{align*}
Moreover, we calculate
\begin{align*}
\Upsilon_{\mathrm{e},\delta}
&=
\int_{\bar{r}_{\delta,-}}^{\bar{r}_{\delta,+}}\Phi^\prime_\delta\at{r}\dint{r}
=
\int_{\bar{r}_{\delta,-}}^{\bar{r}_{\delta,+}}\left(r-\Psi^\prime_\delta\at{r}
-\Psi^\prime_0\at{r}+\Psi^\prime_0\at{r}\right)\dint{r}
\\&
=
\int_{\bar{r}_{\delta,-}}^{\bar{r}_{\delta,+}}\Phi^\prime_0\at{r}\dint{r}
-\int_{\bar{r}_{\delta,-}}^{\bar{r}_{\delta,+}}\left(\Psi^\prime_\delta\at{r}-\Psi^\prime_0\at{r}\right)\dint{r}
\\&
=
\Phi_0\at{\bar{r}_{\delta,+}}-\Phi_0\at{\bar{r}_{\delta,-}}
-2I_\delta
=
\tfrac12\at{\bar{r}_{\delta,+}-1}^2-\tfrac12\at{\bar{r}_{\delta,-}+1}^2
-2I_\delta
\\&
=
\tfrac12\at{\bar{r}_{0,+}-I_\delta-1}^2-\tfrac12\at{\bar{r}_{0,-}-I_\delta+1}^2 - 2I_\delta +\nDO{\delta^2}
\\&
=
\tfrac12\at{\bar{r}_{0,+}-1}^2-\tfrac12\at{\bar{r}_{0,-}+1}^2
-I_\delta\at{\bar{r}_{0,+}-1-\bar{r}_{0,-}-1}
-2I_\delta+\nDO{\delta^2}
\\&=
\Upsilon_{\mathrm{e},0}-I_\delta\at{\bar{r}_{0,+}-\bar{r}_{0,-}}+\nDO{\delta^2}\,.
\end{align*}
Subtracting both results gives $\Upsilon_\delta = \Upsilon_0 +\nDO{\delta^2}$,
the desired result.
\end{proof}
%
%
\section*{Acknowledgement}
%
The authors gratefully acknowledge financial support by the EPSRC (EP/H05023X/1).
%
%
%
%
\def\cprime{$'$} \def\cprime{$'$} \def\cprime{$'$}   \def\polhk#1{\setbox0=\hbox{#1}{\ooalign{\hidewidth   \lower1.5ex\hbox{`}\hidewidth\crcr\unhbox0}}} \def\cprime{$'$}   \def\cprime{$'$}
%

%
%

\begin{thebibliography}{BCS01b}
%
\bibitem[AK91]{Abeyaratne:91a}
%
Rohan Abeyaratne and James~K. Knowles. \newblock Kinetic relations and the propagation of phase boundaries in solids. \newblock {\em Arch. Rational Mech. Anal.}, 114(2):119--154, 1991.
%
\bibitem[BCS01a]{Balk:01a}
%
Alexander~M. Balk, Andrej~V. Cherkaev, and Leonid~I. Slepyan. \newblock Dynamics of chains with non-monotone stress-strain relations. {I}.   {M}odel and numerical experiments. \newblock {\em J. Mech. Phys. Solids}, 49(1):131--148, 2001. %
%
\bibitem[BCS01b]{Balk:01b}
%
Alexander~M. Balk, Andrej~V. Cherkaev, and Leonid~I. Slepyan. \newblock Dynamics of chains with non-monotone stress-strain relations. {I}{I}. {N}onlinear waves and waves of phase transition. \newblock {\em J. Mech. Phys. Solids}, 49(1):149--171, 2001.
%
\bibitem[FP99]{Friesecke:99a}
%
Gero Friesecke and Robert L. Pego. \newblock Solitary waves on {F}{P}{U} lattices. {I}. {Q}ualitative properties,   renormalization and continuum limit. \newblock {\em Nonlinearity}, 12(6):1601--1627, 1999.
%
\bibitem[FW94]{Friesecke:94a}
%
Gero Friesecke and Jonathan A.~D. Wattis. \newblock Existence theorem for solitary waves on lattices. \newblock {\em Comm. Math. Phys.}, 161(2):391--418, 1994.
%
\bibitem[Her10]{Her10a}
%
Michael Herrmann. \newblock Unimodal wavetrains and solitons in convex {F}ermi-{P}asta-{U}lam   chains. \newblock {\em Proc. Roy. Soc. Edinburgh Sect. A}, 140(4):753--785, 2010.
%
\bibitem[Her11]{Herrmann:11b}
%
Michael Herrmann. \newblock Action minimising fronts in general {FPU}-type chains. \newblock {\em J. Nonlinear Sci.}, 21(1):33--55, 2011.
%
\bibitem[HR10]{HR10b}
%
Michael Herrmann and Jens D.~M. Rademacher. \newblock Heteroclinic travelling waves in convex {FPU}-type chains. \newblock {\em SIAM J. Math. Anal.}, 42(4):1483--1504, 2010.
%
\bibitem[HSZ12]{Herrmann:10b}
%
Michael Herrmann, Hartmut Schwetlick, and Johannes Zimmer. \newblock On selection criteria for problems with moving inhomogeneities. \newblock {\em Contin. Mech. Thermodyn.}, 24:21--36, 2012. 
%
\bibitem[Ioo00]{Iooss:00b}
%
G{\'e}rard Iooss. \newblock Travelling waves in the {F}ermi-{P}asta-{U}lam lattice. \newblock {\em Nonlinearity}, 13(3):849--866, 2000.
%
\bibitem[EP05]{EP05}
%
J. M. English and Robert L. Pego. \newblock On the solitary wave pulse in a chain of beads.
\newblock {\em Proc. Amer. Math. Soc.},  133(6): 1763–1768, 2005.
%
\bibitem[Pan05]{Pankov:05a}
%
Alexander Pankov. \newblock {\em Travelling waves and periodic oscillations in   {F}ermi-{P}asta-{U}lam lattices}. \newblock Imperial College Press, London, 2005.
%
\bibitem[SW97]{Smets:97a}
%
Didier Smets and Michel Willem. \newblock Solitary waves with prescribed speed on infinite lattices. \newblock {\em J. Funct. Anal.}, 149(1):266--275, 1997.
%
\bibitem[SZ09]{Schwetlick:07a}
%
Hartmut Schwetlick and Johannes Zimmer. \newblock Existence of dynamic phase transitions in a one-dimensional lattice   model with piecewise quadratic interaction potential. \newblock {\em SIAM J. Math. Anal.}, 41(3):1231--1271, 2009.
%
\bibitem[SZ12]{Schwetlick:08a}
%
Hartmut Schwetlick and Johannes Zimmer. \newblock Kinetic Relations for a Lattice Model of Phase Transitions. \newblock {\em Arch. Ration. Mech.  Anal.}, 206(2):707--724, 2012.
%
\bibitem[SCC05]{SCC05}
%
Leonid I. Slepyan, Andrej Cherkaev, and Elena Cherkaev.
\newblock Transition waves in bistable structures. II. Analytical solution: wave speed and energy dissipation.
\newblock {\em J. Mech. Phys. Solids}, 53(2):407–436, 2005.
%
\bibitem[Sle01]{Sle01}
%
Leonid I. Slepyan.
\newblock Feeding and dissipative waves in fracture and phase transition. I. Some 1D structures and a square-cell lattice.
\newblock {\em J. Mech. Phys. Solids}, 49(3):469–511, 2001.
%
\bibitem[Sle02]{Sle02}
%
Leonid I. Slepyan
\newblock
{\em Models and phenomena in fracture mechanics}.
\newblock Foundations of Engineering Mechanics, Springer-Verlag, Berlin, 2002.
%

\bibitem[Tru87]{Truskinovski:87a}
%
Lev Truskinovsky. \newblock Dynamics of nonequilibrium phase boundaries in a heat conducting   non-linearly elastic medium. \newblock {\em Prikl. Mat. Mekh.}, 51(6):1009--1019, 1987.
%

\bibitem[TV05]{TV05}
%
Lev Truskinovsky and Anna Vanchtein. \newblock Kinetics of martensitic phase transitions: lattice model.
\newblock {\em SIAM J. Appl. Math.}, 66(2):533–553, 2005.
%

\bibitem[Vai10]{Vainchtein:09a}
%
Anna Vainchtein. \newblock The role of spinodal region in the kinetics of lattice phase   transitions. \newblock {\em J. Mech. Phys. Solids}, 58(2):227--240, 2010.
%
%
\end{thebibliography}

\end{document}